\documentclass[numbook,envcountsect,envcountsame,envcountreset,runningheads,smallextended]{svjour3}

\smartqed  
\usepackage{graphicx}
\usepackage{amsfonts,amssymb}
\usepackage{color}

\makeatletter
\newcommand*\bigcdot{\mathpalette\bigcdot@{.5}}
\newcommand*\bigcdot@[2]{\mathbin{\vcenter{\hbox{\scalebox{#2}{$\m@th#1\bullet$}}}}}
\makeatother
\newcommand{\is}{\bigcdot }

\newcommand{\E}{\mathbb E}

\def \Rbrack {[\![}
\def \Lbrack {]\!]}

\def \Rbrack {[\![}
\def \Lbrack {]\!]}

\title{Log-optimal and num\'eraire portfolios for market models stopped at a random time}

\author{ Tahir Choulli \and Sina Yansori}
\institute{Tahir Choulli (corresponding author) and Sina Yansori\at
Department of Mathematical and Statistical Sciences,\\ University of Alberta, Edmonton, Canada \\
              \email{tchoulli@ualberta.ca}  }

\date{Version of July 24th, 2020}

\begin{document}
\titlerunning{Log-related portfolios  under random horizon}
\maketitle






\begin{abstract}
This paper focuses on num\'eraire portfolio and log-optimal portfolio (portfolio with finite expected utility that maximizes the expected logarithm utility from terminal wealth), when a market model $(S,\mathbb F)$ --specified by its assets' price $S$ and its flow of information $\mathbb F$-- is stopped at a random time $\tau$. This setting covers the areas of credit risk and life insurance, where $\tau$ represents the default time and the death time respectively. Thus, the progressive enlargement of $\mathbb F$ with $\tau$, denoted by $\mathbb G$, sounds tailor-fit for modelling the new flow of information that incorporates both $\mathbb F$ and $\tau$. For the resulting stopped model $(S^{\tau},\mathbb G)$, we study the two portfolios in different manners, and describe their computations in terms of the $\mathbb F$-observable parameters of the pair $(S, \tau)$. 
As existence of num\'eraire portfolios for $(S, \mathbb F)$ and $(S^{\tau},\mathbb G)$ is well understood due to \cite{ACDJ1,ACDJ3,CD1,ChoulliDengMa,KardarasKaratzas}, herein we characterize num\'eraire portfolio of $(S^{\tau},\mathbb G)$ in various manners and we single out the types of risks borne by $\tau$ that really affect the portfolio. In contrast to num\'eraire portfolio, both existence and characterization of log-optimal portfolio of $(S^{\tau},\mathbb G)$ pose serious challenges. Among these we mention the following. a) What are the conditions on $\tau$ (preferably in terms of information theoretic concepts) that fully characterize the existence of log-optimal portfolio of $(S^{\tau},\mathbb G)$ when that of $(S,\mathbb F)$ already exists? b) What are the factors that fully determine {\it the increment in maximum expected logarithmic utility from terminal wealth} for the two models $(S^{\tau},\mathbb G)$  and $(S,\mathbb F)$, and how to quantify them?  This problem rises naturally, given that the investor endowed with the flow $\mathbb G$ possesses an informational advantage, as she can see $\tau$ occurring, but also faces various risks. Besides answering deeply these problems and other related challenges, our paper proposes in its details another modelling method for random times and shows their positive and negative impacts on log-optimal portfolio.  \end{abstract}





\section{Introduction}
This paper addresses two portfolios that intimately related to the logarithmic utility. These portfolios, are known in the literature as num\'eraire  and log-optimal portfolios, that we start by defining below for the sake of full precision. To this end, we denote by $W^{\theta}$ the wealth process of the portfolio $\theta$.
\begin{definition}\label{NP/LogOP}  Let $(X, \mathbb H, Q)$ be a market model, where $X$ is the assets' price process, $\mathbb H$ is a filtration, and $Q$ is a probability measure. Consider a fixed investment horizon $T\in(0,+\infty)$, and a portfolio $\theta^*$.\\
{\rm{(a)}} $\theta^*$ is {\it num\'eraire portfolio} for $(X, \mathbb H, Q)$ if  $W^{\theta^*}>0$ and
\begin{eqnarray}\label{NP}
{{W^{\theta}}\over{W^{\theta^*}}}\ \mbox{is a supermartingale under $(\mathbb H, Q)$, any portfolio $\theta$ with $W^{\theta}\geq 0$}.\hskip 0.5cm
\end{eqnarray}
{\rm{(b)}} $\theta^*$ is called {\it log-optimal portfolio} for $(X, \mathbb H, Q)$ if $\theta^*\in \Theta(X,\mathbb H, Q)$ and 
\begin{eqnarray}
 u_T(X,\mathbb H, Q):=\sup_{\theta\in\Theta}E_Q\left[\ln(W^{\theta}_T)\right]= E_Q\left[\ln(W^{\theta^*}_T)\right],\label{LogInfinite}\end{eqnarray}
where $E_Q[.]$ is the expectation under $Q$, and $\Theta:=\Theta(X,\mathbb H, Q)$ is given by 
\begin{eqnarray}\label{AdmissibleSet0}
\hskip -0.6cm \Theta(X,\mathbb H, Q):=\left\{\mbox{portfolio}\ \theta\ \Big|\ W^{\theta}> 0\quad \mbox{and}\quad E_Q\left[\vert \ln(W^{\theta}_T)\vert \right]<+\infty\right\}.\end{eqnarray}
\end{definition}
The problem of maximization of expected logarithm-utility from terminal wealth, defined in (\ref{LogInfinite})-(\ref{AdmissibleSet0}), received a lot of attention in the literature, even though it is a particular case of the utility maximization theory problem. This latter problem is addressed at various levels of generality, and for further details about it we refer to  \cite{CSW,Karatzas,KW99,KZ,merton71,merton73} and the references therein. \\
The num\'eraire portfolio was introduced --up to our knowledge-- in \cite{Long}, where $W^{\theta}/W^{\theta^*}$ is required to be a martingale, while  Definition \ref{NP/LogOP}-(a) goes back to \cite[Definition 4.1]{Becherer}, who remarked that the martingality requirement for $W^{\theta}/W^{\theta^*}$  is too stringent to obtain a general existence result. Then these works was extended and investigated extensively in different directions in \cite{Becherer,ChoulliDengMa,ChristensenLarsen2007,HulleySchweizer,KardarasKaratzas} and the references therein. In \cite{Becherer,ChristensenLarsen2007,HulleySchweizer,GollKallsen}, it was proved that under no-free-lunch-with-vanishing-risk assumption (NFLVR hereafter) and/or $\sup_{\theta\in\Theta}E\left[\ln(W^{\theta}_T)\right]<+\infty$, the two portfolios (log-optimal and num\'eraire) coincide. Using the change of probability technique, deep and precise connection between the two portfolios is established in \cite{ChoulliDengMa} under no assumption. Furthermore, very recently in \cite{ChoulliYansori2} and under no assumption at all, this connection is elaborated without changing the probability and the explicit computation of log-optimal portfolio and other related properties were also developed. This latter work, for which we also refer for more detailed discussions about the literature on these portfolios, definitely seals these questions about the two portfolios, and is vital for our current paper as it assumes no assumption.
  \subsection{What are our objectives and what does the literature say about them?} 
  In this paper, we consider an initial market model  represented by the pair $(S,\mathbb F)$, where $S$ represents the discounted stock prices for $d$-stocks, and $\mathbb F$ is the ``public" information that is available to all agents. To this initial market model, we add a random time $\tau$ that might not be seen through $\mathbb F$ when it occurs (mathematically speaking $\tau$ might not be an $\mathbb F$-stopping time). In this context, we adopt the progressive enlargement of filtration to model the larger information that includes both $\mathbb F$ and $\tau$. The obtained new informational system,  that we denote by $(S^{\tau},\mathbb G)$, allows us to keep in mind credit risk theory and life insurance as potential applications of our results, besides the general financial setting of markets with random horizons. For this informational market, our ultimate goal lies in measuring the impact of $\tau$ on num\'eraire and log-optimal portfolios,  no matter what is the model for $(S,\mathbb F)$ and no matter how it is related to $\tau$ that is an arbitrary random time with positive ``survival probability" (i.e. Az\'ema supermartingale). Our setting falls into the vague topic of {\it portfolio problem under asymmetries of information}. The mathematical literature on information modelling proposes only two cases of incorporating the extra information depending on whether this information is added at the beginning of the investment interval or progressively over time.
   
  The first case corresponds mathematically  to the initial enlargement of filtration and is known in the finance and mathematical finance literatures as {\it the insider trading} setting. For this insider framework, log-optimal portfolio  is extensively studied and we refer the reader to \cite{amendingerimkellerschweizer98,ADImkeller,AImkeller,JImkellerKN,GrorudPontier,pikovskykaratzas96,kohatsusulem06}  and the references therein to cite few. Most of this literature focuses on two intimately related questions on log-optimal portfolio for the model  $(S,{\mathbb G}^*)$, where ${\mathbb G}^*$ is the initial enlargement of $\mathbb F$ with a random variable $L$ that represents the extra knowledge. In fact, under some assumption on the pair $(L, \mathbb F)$, frequently called Jacod's assumption, the existence of log-optimal portfolio and the evaluation of  the {\bf  increment of expected logarithm-utility from terminal wealth} (denoted hereafter by IEU$_{log}(S,{\mathbb G}^*, \mathbb F)$) for both models $(S,{\mathbb G}^*)$ and $(S,\mathbb F)$ represent the core contribution of these papers, where it is proven that  
\begin{eqnarray}\label{InsiderFormula}
\mbox{IEU}_{log}(S,{\mathbb G}^*, \mathbb F):=u_T(S, {\mathbb G}^*)-u_T(S, \mathbb F)=\mbox{relative entropy}(P\big| Q^*).
\end{eqnarray}
Hence, in this insider setting, log-optimal portfolio for $(S,{\mathbb G}^*)$  exists if and only if $P$ has a finite entropy with respect to $Q^*$ , a precise probability measure associated to $L$  that is explicitly described. In particular, the quantity  $\mbox{IEU}_{log}(S,{\mathbb G}^*, \mathbb F)$ is always a true gain due to the advantage of knowing fully $L$ by the investor endowed with the flow $\mathbb G^*$. The formula (\ref{InsiderFormula}) was initially derived in \cite{pikovskykaratzas96} for the Brownian filtration, and it was extended to models driven by general continuous local martingales in \cite{amendingerimkellerschweizer98}, where the authors connect this formula with Shannon entropy of $L$ for some models. The Shannon concept was exploited deeply  in \cite{ADImkeller} afterwards, where the authors show its important role in measuring the impact of inside-information on log-optimal portfolio.
  
The second case of information modelling, which suggests to add the extra information over time as it occurs, leads to the progressive enlargement filtration, and is tailor-fit to our current financial setting in contrast to the initial enlargement. Our economic and financial problem that deals with how a random horizon will impact an investment (in particular num\'eraire and log-optimal portfolios) can be traced back to Fisher \cite{fisher1931}. Since then, the problem has been addressed in the economic literature by focusing on discrete market models and the impact of the distribution of $\tau$ only, see \cite{Hakansson,Yaari1965} and the references therein. Thus, our paper seems to be the {\it first of its kind} in virtue of the general setting for $(S, \mathbb F, \tau)$ and both the qualitative and quantitative results obtained. Below, we highlight the intuitive ideas that leaded to these results.

Given that a log-optimal portfolio is a num\'eraire portfolio, see \cite{ChoulliYansori2} and the references therein, we start by addressing num\'eraire portfolio of $(S^{\tau},\mathbb G)$.  Thanks to \cite{ChoulliDengMa,KardarasKaratzas} that connects the existence of num\'eraire portfolio to the concept of No-Unbounded-Profit-with-bounded-risk (NUPBR hereafter), and the recent works \cite{ACDJ1,ACDJ3,CD1} on NUPBR for the stopped model $(S^{\tau},\mathbb G)$, the problem of existence of num\'eraire portfolio for $(S^{\tau},\mathbb G)$ is completely understood. Thus, herein, we focus on describing this num\'eraire portfolio in terms of the $\mathbb F$-observable data and processes, and mainly single out 
\begin{eqnarray}\label{Q1}
\mbox{which types of risks borne by $\tau$ that affect num\'eraire portfolio.}\hskip 0.65cm\end{eqnarray}
For log-optimal portfolio of $(S^{\tau}, \mathbb G)$, the situation is more challenging, and the problem of its existence is the first obstacle. To address this, we appeal to the explicit description of the set of deflators for $(S^{\tau}, \mathbb G)$, recently developed together with its application to NFLVR in \cite{ChoulliYansori1}, and answer the following. 
\begin{eqnarray}\label{Q2}
\mbox{For which models of $(S, \tau)$, log-optimal portfolio of $(S^{\tau}, \mathbb G)$ exists?}\end{eqnarray}
It is worth mentioning that this existence question is much deeper and general than the corresponding one addressed in the insider setting. Indeed, in our framework, there is no hope for (\ref{InsiderFormula}) to hold in its current form, and only a practical answer to (\ref{Q2}) will allow us to answer the question below.
\begin{eqnarray}\label{Q3}
&&\mbox{What {\it informational condition} on $\tau$  that characterizes the existence }\nonumber\\ 
&&\mbox{ of log-optimal portfolio for $(S^{\tau}, \mathbb G)$ if $(S,\mathbb F)$ has log-optimal portfolio?}\hskip 0.75cm\end{eqnarray}
For our case of random horizon, the {\it increment of expected logarithm-utility} between $(S^{\tau}, \mathbb G)$ and $(S,\mathbb F)$, that we denote by IEU$_{log}(S, \tau, \mathbb F)$, is defined by 
\begin{eqnarray}\label{Delta(S, Tau)}
\mbox{IEU}_{log}(S,\tau, \mathbb F):=\Delta_T(S, \tau, \mathbb F):=u_T(S^{\tau}, \mathbb G)-u_T(S, \mathbb F),\end{eqnarray}
and is affected by many factors, and hence we address the question of 
\begin{eqnarray}\label{Q5}
\mbox{which factors that explain how sensitive IEU$_{log}(S,\tau,\mathbb F)$ to $\tau$}?\end{eqnarray}
To answer this question, we prefer study the much deeper question below, that deals with the explicit computation of log-optimal portfolio.
\begin{eqnarray}\label{Q4}
\mbox{How log-optimal portfolio of $(S^{\tau},\mathbb G)$ can be described using $\mathbb F$ only?}\hskip 0.75cm\end{eqnarray}

\subsection{What are our achievements?} Our mathematical and financial achievements are numerous and highly novel in both conceptual and methodological aspects. In fact, we answer all the aforementioned questions above (i.e. (\ref{Q1}), (\ref {Q2}),(\ref{Q3}), (\ref{Q5}), (\ref{Q4})) and other related problems in a very detailed and deep analysis and in various manners. In fact, we describe log-optimal portfolio, the structures of its associated log-optimal deflator, and num\'eraire portfolio for $(S^{\tau},\mathbb G)$ in different manners. As a result, we prove that the random horizon induces randomness in agent's preference (or the agent's impatience as called in Fisher \cite{fisher1931}). This connects the random horizon issue to random utilities that appeared in economics within the {\it random utility model theory} due to the psychometric literature that gave empirical evidence for stochastic choice behaviour. For details about this theme, we refer to \cite{Clark96,Cohen80,Mcfadden90,Suppes89} and for applications we refer to \cite{ChoulliMaMorlais,ChoulliMa,Ma,MusielaZariphoupoulou,KZ}.\\
Our results show that both portfolios (num\'eraire and log-optimal) are affected by the correlation between $S$ and $\tau$ only, and this correlation is explicitly parametrized using $\mathbb F$-adapted processes. We prove that $\Delta_T(S,\tau,\mathbb F)$, defined in (\ref{Delta(S, Tau)}), depends on four factors. These factors, that we quantify explicitly, are ``the cost-of-leaving-earlier", ``the information-premium", which is due to the advantage of knowing the occurring of $\tau$ when it happens, ``the correlation" between num\'eraire portfolio of $(S,\mathbb F)$ and $\tau$, and ``the correlation-risk" between $\tau$ and $S$. These factors explain how complex is the impact of a random horizon on portfolio compared to the impact of an {\it inside-information}. 

This paper contains five sections including the current one. Section  \ref{section2} presents the mathematical and the financial model besides the corresponding required notation and some preliminaries that state some existing results that are important herein. Section \ref{section3} addresses num\'eraire portfolio for $(S^{\tau},\mathbb G)$, while Section \ref{section4} focuses on the existence of log-optimal portfolio and the duality. Section \ref{section5} describes explicitly both num\'eraire and log-optimal portfolios using the $\mathbb F$-predictable characteristics of the model, and discusses its financial applications and consequences. The paper contains an appendix where some proofs are relegated and  some technical (new and existing) results are detailed. 
\section{The mathematical model and preliminaries}\label{section2}
Throughout the paper, by  ${\mathbb H}$ we denote an arbitrary  filtration that satisfies the usual conditions of completeness and right continuity.  For any process $X$, the $\mathbb H$-optional projection and dual $\mathbb H$-optional projection of $X$, when they exist, will be denoted by $^{o,\mathbb H}X$ and $X^{o,\mathbb H}$ respectively.  Similarly, we denote by $^{p,\mathbb H}X$ and $X^{p,\mathbb H}$ the $\mathbb H$-predictable projection and dual predictable projection of $X$ when they exist. The set ${\cal M}(\mathbb H, Q)$ denotes the set of all $\mathbb H$-martingales under $Q$, while ${\cal A}(\mathbb H, Q)$ denotes the set of all  optional processes with integrable variation under $Q$. When there is  no risk of confusion, we simply omit the probability for the sake of simplifying notation.  For an $\mathbb H$-semimartingale $X$, by $L(X,\mathbb H)$ we denote the set of $\mathbb H$-predictable processes that are $X$-integrable in the semimartingale sense.  For $\varphi\in L(X,\mathbb H)$, the resulting integral of $\varphi$ with respect to $X$ is denoted by $\varphi\is X$. For $\mathbb H$-local martingale $M$, we denote by $L^1_{loc}(M,\mathbb H)$ the set of $\mathbb H$-predictable processes $\varphi$ that are $X$-integrable and the resulting integral $\varphi\is M$ is an $\mathbb H$-local martingale. If ${\cal C}(\mathbb H)$ is the set of processes that are adapted to $\mathbb H$, then ${\cal C}_{loc}(\mathbb H)$ is the set of processes, $X$, for which there exists a sequence of $\mathbb H$-stopping times, $(T_n)_{n\geq 1}$, that increases to infinity and $X^{T_n}$ belongs to ${\cal C}(\mathbb H)$, for each $n\geq 1$. For any $\mathbb H$-semimartinagle, $L$, we denote by ${\cal E}(L)$ the Doleans-Dade (stochastic) exponential, it is the unique solution to the stochastic differential equation $dX=X_{-}dL,\quad X_0=1,$  given by
$$ {\cal E}_t(L)=\exp(L_t-{1\over{2}}\langle L^c\rangle_t)\prod_{0<s\leq t}(1+\Delta L_s)e^{-\Delta L_s}.$$
{\bf How our financial model is parametrized?} Our model starts with a filtered probability space $\left(\Omega, {\cal F}, \mathbb F,P\right)$. Here the filtration $\mathbb F:=({\cal F}_t)_{t\geq 0}$, which represents the ``public" flow of information available to all agent over time, satisfies the usual conditions of right continuity and completeness. On this stochastic basis, we suppose given a d-dimensional $\mathbb F$-semimartingale, $S$, that models the discounted price process of d risky assets. In addition to this initial market model $(S, \mathbb F)$, we consider a random time $\tau$, that might represent the death time of an agent or the default time of a firm, and hence it might not be an $\mathbb F$-stopping time in general. To this random time, we associate the non-decreasing process $D$ and the filtration $\mathbb G:=({\cal G}_t)_{t\geq 0}$ given by
\begin{equation}\label{processD}
D:=I_{\Rbrack\tau,+\infty\Rbrack},\ \ \ \ {\cal G}_t:={\cal G}^0_{t+}\ \ \mbox{where} \ \ {\cal G}_t^0:={\cal F}_t\vee\sigma\left(D_s,\ s\leq t\right).
\end{equation}
It is clear that $\mathbb G$ makes $\tau$ a stopping time. In fact, it is the smallest filtration, satisfying the usual conditions, that makes $\tau$ a stopping time and contains $\mathbb F$. It is the progressive enlargement of $\mathbb F$ with $\tau$.  Besides $D$ and $\mathbb G$, other $\mathbb F$-adapted processes intimately related to $\tau$ play central roles in our analysis. Among these, the following survival probabilities, also called Az\'ema supermartingales in the literature, and are given by 
\begin{eqnarray}\label{GGtilde}
G_t :=^{o,\mathbb F}(I_{\Rbrack0,\tau\Rbrack})_t= P(\tau > t | {\cal F}_t) \ \mbox{ and } \ \widetilde{G}_t :=^{o,\mathbb F}(I_{\Rbrack0,\tau\Lbrack})_t= P(\tau \ge t | {\cal F}_t),\hskip 0.5cm\end{eqnarray}
while the process
\begin{equation} \label{processm}
m := G + D^{o,\mathbb F},
\end{equation}
is an $\mathbb F$-martingale.  Then  thanks to  \cite{ACJ} and  \cite{ChoulliDavelooseVanmaele}, we claim the following.
\begin{theorem}
The following assertions hold.\\
{\rm{(a)}} For any  $M\in{\cal M}_{loc}(\mathbb F)$, the process
\begin{equation} \label{processMhat}
{\cal T}(M) := M^\tau -{\widetilde{G}}^{-1} I_{\Lbrack 0,\tau\Lbrack} \is [M,m] +  I_{\Lbrack 0,\tau\Lbrack} \is\Big(\sum \Delta M I_{\{\widetilde G=0<G_{-}\}}\Big)^{p,\mathbb F},\end{equation}
is a $\mathbb G$-local martingale.\\
 {\rm{(b)}}We always have 
\begin{equation} \label{processNG}
N^{\mathbb G}:=D - \widetilde{G}^{-1} I_{\Lbrack 0,\tau\Lbrack} \is D^{o,\mathbb  F}\in{\cal M}(\mathbb G)\cap{\cal A}(\mathbb G),
\end{equation}
and $H\is N^{\mathbb G}\in {\cal M}_{loc}(\mathbb G)\cap{\cal A}_{loc}(\mathbb G)$ for any $H$ belonging to
\begin{equation} \label{SpaceLNG}
{\mathcal{I}}^o_{loc}(N^{\mathbb G},\mathbb G) := \Big\{K\in \mathcal{O}(\mathbb F)\ \ \big|\quad \vert{K}\vert G{\widetilde G}^{-1} I_{\{\widetilde{G}>0\}}\is D\in{\cal A}^+_{loc}(\mathbb G)\Big\}.
\end{equation}
\end{theorem}
For $p\in [1,+\infty)$ and a $\sigma$-algebra ${\cal H}$ on $\Omega\times [0,+\infty[$,  we define 
$L^p_{loc}\left({\cal H}, P\otimes D\right)$ as the set of all processes $X$  for which there exists a sequence of $\mathbb F$-stopping times $(T_n)_{n\geq 1}$ that increases to infinity almost surely and $X^{T_n}$ belongs to $L^p\left({\cal H}, P\otimes D\right)$ given by 
\begin{equation}\label{L1(PandD)Local}
L^p\left({\cal H}, P\otimes D\right):=\left\{ X\ {\cal H}\mbox{-measurable}\big|\ \E[\vert X_{\tau}\vert^p I_{\{\tau<+\infty\}}]<+\infty\right\}.\end{equation}

The explicit description of the set of all deflators for $(S^{\tau},\mathbb G)$ is vital  for our analysis of log-optimal and num\'eraire portfolios undertaken in the coming sections. Thus, we start by  recalling the mathematical definition of delators.
\begin{definition}\label{DeflatorDefinition} Let $X$ be an $\mathbb H$-semimartingale and $Z$ be a process.\\
  We call $Z$ a deflator for $(X,\mathbb H)$ if $Z>0$ and $Z{\cal E}(\varphi\is X)$ is an $\mathbb H$-supermartingale, for any $\varphi\in L(X, \mathbb H)$ such that $\varphi\Delta X\geq -1$. \\ The set of all deflators for $(X,\mathbb H)$ will be denoted by ${\cal D}(X,\mathbb H)$.
 \end{definition}
Throughout the paper, the following subset of ${\cal D}(X,\mathbb H)$ will be very useful 
\begin{eqnarray}
{\cal D}_{log}(X,\mathbb H)&&:=\Bigl\{Z\in {\cal D}(X,\mathbb H)\ \big| E[-\ln(Z_T)]<+\infty\Bigr\}.\label{DeflatorsLOG}
\end{eqnarray}
The following is borrowed from \cite{ChoulliYansori1}, and parametrizes explicitly ${\cal D}(S^{\tau}, \mathbb G)$. 
\begin{theorem}\label{GeneralDeflators} 
Suppose $G > 0$, and let ${\cal T}(\cdot)$ be the operator defined in (\ref{processMhat}). Then the following assertions hold.\\
{\rm{(a)}} $Z^{\mathbb G}$ is a deflator for $(S^{\tau}, \mathbb G)$  (i.e. $Z^{\mathbb G}\in {\cal D}(S^{\tau}, \mathbb G)$) if and only if there exists unique $\left(Z^{\mathbb F}, \varphi^{(o)}, \varphi^{(pr)}\right)$ such that $Z^{\mathbb F}\in{\cal D}(S, \mathbb F)$, $(\varphi^{(o)},\varphi^{(pr)})$ belongs to ${\cal I}^o_{loc}(N^{\mathbb G},\mathbb G)\times L^1_{loc}({\rm{Prog}}(\mathbb F),P\otimes D)$,  
  \begin{eqnarray}
&&\varphi^{(pr)}>-1,\quad 
-{\widetilde G}/ G<\varphi^{(o)},\ \varphi^{(o)}(\widetilde G -G)<{\widetilde G},\ P\otimes D\mbox{-a.e.,} \label{ineqMultiGeneral1}\\
&&\mbox{and}\quad Z^{\mathbb G}={{(Z^{\mathbb F})^{\tau}}\over{{\cal E}(G_{-}^{-1}\is m)^{\tau}}}{\cal E}(\varphi^{(o)}\is N^{\mathbb G}){\cal E}(\varphi^{(pr)}\is D).\label{repKGMultiGEneral}\end{eqnarray}
{\rm{(b)}}  For any $K\in {\cal M}_{0,loc}(\mathbb F)$, we always have
\begin{eqnarray*}
{{{\cal E}(K)^{\tau}/{\cal E}(G_{-}^{-1}\is m)^{\tau}}}={\cal E}\Bigl({\cal T}(K-G_{-}^{-1}\is m)\Bigr)\in {\cal M}_{loc}(\mathbb G).\end{eqnarray*}
{\rm{(c)}} The process $1/ { \cal E}(G_{-}^{-1}\is m)^{\tau}$ is a $\mathbb G$-martingale, and for any $T\in (0,+\infty)$ we denote by ${\widetilde Q}_T$ the probability measure is given by
\begin{eqnarray}\label{Qtilde}
d{\widetilde Q}_T:=\Bigl({ \cal E}_{T\wedge\tau}\left(G_{-}^{-1}\is m\right)\Bigr)^{-1}dP.\end{eqnarray}
\end{theorem}

\section{Num\'eraire portfolio under random horizon}\label{section3}
This section addresses the impact of $\tau$ on the num\'eraire portfolio.  To this end, we start by giving a mathematical sense to Definition \ref{NP/LogOP} as follows. 
\begin{definition}\label{Math4Definition1.1} Let $(X,\mathbb H, Q)$ be a market model, where $\mathbb H$ is a filtration, $Q$ is a probability measure, and $X$ is an $\mathbb H$-semimartingale under $Q$. \\
{\rm{(a)}} A portfolio $\theta$ is a predictable process that is $X$-integrable (i.e. $\theta\in L(X, Q,\mathbb H)$). A wealth process $W^{\theta}$ associated to the pair $(\theta, x)$ of a portfolio and an initial capital (i.e. $x>0$) is given by 
\begin{eqnarray}\label{WealthProcess}
W^{\theta}:=x+\theta\is X.\end{eqnarray}
{\rm{(b)}} Let $\theta$ be a portfolio and $x>0$ be an initial capital. If the wealth process for the pair $(\theta, x)$ satisfies $W^{\theta}>0$ and $W^{\theta}_{-}>0$, then the process
\begin{eqnarray}\label{PrtfolioRate}
\varphi^{(\theta)}:={{\theta/W^{\theta}_{-}}}\quad \mbox{ is called {\it portfolio rate}},\end{eqnarray}
\end{definition}
 
 \begin{remark}\label{Remark4Definitions}
 {\rm{(a)}} It is important to remark that the portfolio rate $\varphi^{(\theta)}$ does not depend on the initial capital $x$ and it depends on the portfolio $\theta$ only. Furthermore, the triplet $(\theta, x, \varphi^{(\theta)})$ also satisfies $W^{\theta}=x{\cal E}(\varphi^{(\theta)}\is X)$.\\
  {\rm{(b)}} If ${\cal D}(X,\mathbb H)\not=\emptyset$, then for any pair $(\theta, x)$ with $W^{\theta}>0$ we also have $W^{\theta}_{-}>0$ and the portfolio rate exists. Thus, in this case, there is no loss of generality in assuming $x=1$ when addressing the problem (\ref{LogInfinite}). \\
   {\rm{(c)}} By comparing Definitions \ref{DeflatorDefinition} and \ref{NP/LogOP}, it is clear that, if the num\'eraire portfolio rate $\widetilde\phi$ for $(X,\mathbb H)$ exists, then ${\widetilde Z}:=1/{\cal E}(\widetilde\phi\is X)$ belongs to ${\cal D}(X,\mathbb H)$.\end{remark}
Below, we elaborate the principal  result of this section. 

\begin{theorem}\label{NumeraireGeneral} Suppose that $G>0$. Then the following assertions hold.\\
{\rm{(a)}} The num\'eraire portfolio for $(S,\mathbb F)$ exists if and only if the num\'eraire portfolio for $(S^{\tau}, \mathbb G)$ exists also.\\
{\rm{(b)}}  If $\widetilde\varphi$ is num\'eraire portfolio rate for $(S,\mathbb F)$, then  ${\widetilde\varphi}I_{\Lbrack0,\tau\Lbrack}$ is num\'eraire portfolio rate for $(S^{T\wedge\tau},\mathbb G, {\widetilde Q}_T)$, for any $T\in (0,+\infty)$, where ${\widetilde Q}_T$ is given by (\ref{Qtilde}).
{\rm{(c)}} If ${\widetilde\varphi}^{\mathbb G}$ is the num\'eraire portfolio rate for $(S^{\tau},\mathbb G)$, then $^{p,\mathbb F}({\widetilde\varphi}^{\mathbb G}I_{\Lbrack0,\tau\Lbrack})/G_{-}$ is num\'eraire portfolio rate for $(S^{\sigma}, \mathbb F, {\widehat Q}_{\sigma})$, for any $\mathbb F$-stopping time $\sigma$ such that ${\cal E}(G_{-}^{-1}\is m)^{\sigma}$ is martingale, where
\begin{eqnarray}\label{Qsigma}
d{\widehat Q}_{\sigma}:={\cal E}_{\sigma}(G_{-}^{-1}\is m)dP.\end{eqnarray}
\end{theorem}
Herein, we discuss some of the ingredients of the theorem and importantly its meaning and contributions, while its proof will be given afterwards. In virtue of \cite[Proposition 3.4 and Theorem 4.2]{ChoulliYansori1} (applied to the constant process $S\equiv 1$), the process ${\cal E}(-G_{-}^{-1}\is {\cal T}(m))=1/{\cal E}(G_{-}^{-1}\is m)^{\tau}$ is a $\mathbb G$-martingale, and hence ${\widetilde Q}_T$ is a well defined probability for any $T\in (0,+\infty)$. Thanks to a combination of \cite[Theorem 2.8]{ChoulliDengMa} (see also \cite{KardarasKaratzas}) with \cite[Theorem 2.15]{ACDJ1} and \cite[Theorem 2.4 or 2.7]{ACDJ3}, it is clear that under the condition $G>0$, the proof of assertion (a) follows immediately.  Hence, the principal contribution of our theorem lies in describing precisely and explicitly how num\'eraire portfolio for $(S^{\tau}, Q^{\mathbb G}, \mathbb G)$ can be obtained from num\'eraire portfolio of $(S, Q^{\mathbb F}, \mathbb F)$  and vice-versa, where $Q^{\mathbb G}$ and $Q^{\mathbb F}$ are probabilities on ${\cal G}_T$ and ${\cal F}_T$ respectively, that quantify stochastically some how the {\it correlated risks} borne by $\tau$. 
 \begin{proof}{\it of Theorem \ref{NumeraireGeneral}.} In virtue of the above discussion, this proof deals with assertions (b) and (c) only, and it will be given in two parts.\\
{\bf Part 1.} Here, we prove assertion (b). Suppose that  num\'eraire portfolio  for $(S,\mathbb F)$ exists, and denote by  $\widetilde\varphi$ its num\'eraire portfolio rate. Thus, for any $\varphi\in {\cal L}(S, \mathbb F)\cap L(S, \mathbb F)$, the process
\begin{eqnarray*}
X:={{{\cal E}(\varphi\is S)/{\cal E}(\widetilde\varphi\is S)}}\quad\mbox{is an $\mathbb F$-supermartingale}.\end{eqnarray*}
Hence, in virtue of Proposition \ref{Hzero4Log(Z)} -(a), there exist unique $M\in {\cal M}_{loc}(\mathbb F)$ and a nondecreasing and $\mathbb F$-predictable process $V$ such that $X={\cal E}(M)\exp(-V)$. Therefore, due to Theorem \ref{GeneralDeflators} -(b) (see also \cite[Proposition 3.4]{ChoulliYansori1}) that states that ${\cal E}(M)^{\tau}/{\cal E}(G_{-}^{-1}\is m)^{\tau}$ is a $\mathbb G$-local martingale, we deduce that $X^{\tau}/{\cal E}(G_{-}^{-1}\is m)^{\tau}$ is $\mathbb G$-supermartinagle, or equivalently $X^{\tau\wedge T}$ is $\mathbb G$-supermartingale under ${\widetilde Q}_T$ given in (\ref{Qtilde}), for any $T\in (0,+\infty)$. As $\varphi$ spans the set $ {\cal L}(S, \mathbb F)\cap L(S, \mathbb F)$, the proof of assertion (b) follows from Lemma \ref{PortfolioGtoF}-(b) and (c).\\
{\bf Part 2.} This part proves assertion (c). Suppose that  num\'eraire portfolio for $(S^{\tau},\mathbb G)$ exists, and denote by $\widetilde\varphi^{\mathbb G}$ its portfolio rate. Then put 
\begin{eqnarray*}
\varphi^{\mathbb F}:=^{p,\mathbb F}(\widetilde\varphi^{\mathbb G} I_{\Lbrack0,\tau\Lbrack})/G_{-},\end{eqnarray*}
and remark that $\varphi^{\mathbb F}=\widetilde\varphi^{\mathbb G} $ on $\Lbrack0,\tau\Lbrack$, and $\varphi^{\mathbb F}\in {\cal L}(S,\mathbb F)\cap L(S,\mathbb  F)$ due to Lemma \ref{PortfolioGtoF}-(b) and (c). Then for any $\varphi\in {\cal L}(S,\mathbb F)\cap L(S,\mathbb  F)$, the process ${\cal E}(\varphi\is S)^{\tau}/{\cal E}(\varphi^{\mathbb F}\is S)^{\tau}$ is a $\mathbb G$-supermartingale. As a result, for an $\mathbb F$-stopping time $\sigma$ such that ${\cal E}(G_{-}^{-1}\is m)^{\sigma}$ is martingale and for any $0\leq s\leq t\leq\sigma$, we have 
\begin{eqnarray*}
E\left[{{{\cal E}_{t\wedge\tau}(\varphi\is S)}\over{{\cal E}_{t\wedge\tau}(\varphi^{\mathbb F}\is S)}}\Big|\ {\cal G}_s\right]I_{\{\tau>s\}}\leq {{{\cal E}_{s}(\varphi\is S)}\over{{\cal E}_{s}(\varphi^{\mathbb F}\is S)}}I_{\{\tau>s\}}.
\end{eqnarray*}
By taking conditional expectation with respect to ${\cal F}_s$ on both sides, we obtain 
\begin{eqnarray}\label{Inequality200}
E\left[{{{\cal E}_{t\wedge\tau}(\varphi\is S)}\over{{\cal E}_{t\wedge\tau}(\varphi^{\mathbb F}\is S)}}I_{\{\tau>s\}}\Big|\ {\cal F}_s\right]\leq G_s{{{\cal E}_{s}(\varphi\is S)}\over{{\cal E}_{s}(\varphi^{\mathbb F}\is S)}}.
\end{eqnarray}
It is clear that we always have
\begin{eqnarray*}
{{{\cal E}_{t\wedge\tau}(\varphi\is S)}\over{{\cal E}_{t\wedge\tau}(\varphi^{\mathbb F}\is S)}}I_{\{\tau>s\}}
&&={{{\cal E}_{t}(\varphi\is S)}\over{{\cal E}_{t}(\varphi^{\mathbb F}\is S)}}I_{\{\tau>t\}}+\int_s^t{{{\cal E}_{u}(\varphi\is S)}\over{{\cal E}_{u}(\varphi^{\mathbb F}\is S)}}dD_u.\end{eqnarray*}
By inserting this and $G=G_0{\cal E}(-{\widetilde G}^{-1}\is D^{o,\mathbb F}){\cal E}(G_{-}^{-1}\is m)$ in (\ref{Inequality200}), we derive 
\begin{eqnarray*}
&&E_{{\widehat Q}_{\sigma}}\left[{{{\cal E}_{t}(\varphi\is S)}\over{{\cal E}_{t}(\varphi^{\mathbb F}\is S)}}{\cal E}_t(-{1\over{\widetilde G}}\is D^{o,\mathbb F})+\int_s^t {{{\cal E}_{u}(\varphi\is S)}\over{{\cal E}_{u}(\varphi^{\mathbb F}\is S)}}{\cal E}_u(-{1\over{\widetilde G}}\is D^{o,\mathbb F}){{dD^{o,\mathbb F}_u}\over{G_u}}\Big|\ {\cal F}_s\right]\\
&&\leq {{{\cal E}_{s}(\varphi\is S)}\over{{\cal E}_{s}(\varphi^{\mathbb F}\is S)}}{\cal E}_s(-{\widetilde G}^{-1}\is D^{o,\mathbb F})
\end{eqnarray*}
Then put $X_u:= {\cal E}_{u\wedge\sigma}(\varphi\is S)/{\cal E}_{u\wedge\sigma}(\varphi^{\mathbb F}\is S)$ for $u\geq 0$, and deduce that the above inequality is equivalent to the fact that 
\begin{eqnarray*}
X{\cal E}(-{1\over{\widetilde G}}\is D^{o,\mathbb F})^{\sigma}+\left({{X}\over{G}}{\cal E}_u(-{1\over{\widetilde G}}\is D^{o,\mathbb F})\is D^{o,\mathbb F}\right)^{\sigma}\ \mbox{is an $(\mathbb F, \widehat Q_{\sigma})$-supermartingale.}\end{eqnarray*}
By combining ${\cal E}(-{\widetilde G}^{-1}\is D^{o,\mathbb F})/G={\cal E}_{-}(-{\widetilde G}^{-1}\is D^{o,\mathbb F})/{\widetilde G}$ with integration by part formula, we deduce that the above fact is equivalent to ${\cal E}_{-}(-{\widetilde G}^{-1}\is D^{o,\mathbb F})\is X$ being an $\mathbb F$-supermartingale under $ \widehat Q_{\sigma}$. Therefore, we conclude that $X$ is a nonnegative $\mathbb F$-local supermartingale under $ \widehat Q_{\sigma}$, and hence assertion (c) follows immediately. This ends the proof of the theorem.\qed
\end{proof}
\begin{corollary}\label{ParticularCases}  The following assertions hold.\\
{\rm{(a)}} Suppose ${\cal E}(G_{-}^{-1}\is m)$ is a martingale. Then the num\'eraire portfolio for $(S^{\tau},\mathbb G)$ exists if and only if for any  $T\in (0,+\infty)$, the num\'eraire portfolio for $(S^T,\mathbb F, {\widehat Q}_T)$ exists, where  ${\widehat Q}_T$ is given by (\ref{Qsigma}). Furthermore, both portfolios coincide on $\Lbrack0,\tau\wedge T\Lbrack$. \\
{\rm{(b)}} Suppose $\tau$ is a pseudo-stopping time (i.e. $E[M_{\tau}]=E[M_0]$ for any bounded $\mathbb F$-martingale). Then num\'eraire portfolio for $(S^{\tau},\mathbb G)$ exists if and only if num\'eraire portfolio for $(S,\mathbb F)$ exists, and both portfolios coincide on $\Lbrack0,\tau\Lbrack$.
\end{corollary}
\begin{proof} Assertion (a) is a direct consequences of  Theorem \ref{NumeraireGeneral}, while assertion (b) follows from combining assertion (a) and the fact that when $\tau$ is a pseudo-stopping time then $m\equiv m_0$ and hence ${\cal E}(G_{-}^{-1}\is m)=1$ and  ${\widehat Q}_T=P$ on ${\cal F}_T$ for any $T\in(0,+\infty)$. This ends the proof of the corollary.\qed
\end{proof}
\section{ Log-optimal portfolio for $(S^{\tau},\mathbb G)$: Existence and duality}\label{section4}
This section quantifies the impact of $\tau$ on the existence of log-optimal portfolio and hence answers (\ref{Q2}), and it is not technical at all. 
In virtue of \cite[Theorem 1]{ChoulliYansori2}, see Theorem \ref{LemmaCrucial} in the appendix, the practical and easy manner to deal with log-optimal portfolio is to look at the dual set, i.e. the set ${\cal D}_{log}(S^{\tau},\mathbb G)$ given by (\ref{DeflatorsLOG}), or equivalently to address the dual minimization problem 
 \begin{eqnarray}\label{dualproblem}
\min_{Z\in {\cal D}_{log}(S^{\tau},\mathbb G)}E\left[-\ln(Z_T)\right].
\end{eqnarray}
 We start by defining Hellinger processes, previously defined in \cite{ChoulliStricker2005,ChoulliStricker2007}. These processes appear naturally in quantifying the information borne in $\tau$. 
 \begin{definition}\label{Hellinger}
 Let $N$ be an $\mathbb H$-local martingale such that $1+\Delta N>0$.\\
 1) We call a Hellinger process of order zero for $N$, denoted by $h^{(0)}(N,\mathbb H)$, the process $ h^{(0)}(N,\mathbb H):=\left( H^{(0)}(N,\mathbb H)\right)^{p,\mathbb H}$ when this projection exists,  where 
 \begin{eqnarray}\label{HellingerLog}
 H^{(0)}(N,\mathbb H):={1\over{2}}\langle N^c\rangle^{\mathbb H}+\sum\left(\Delta N-\ln(1+\Delta N)\right).\end{eqnarray}
2) We call an entropy-Hellinger process for $N$, denoted by $h^{(E)}(N,\mathbb H)$, the process $h^{(E)}(N,\mathbb H):=\left( H^{(E)}(N,\mathbb H) \right)^{p,\mathbb H}$ when this projection exists, where 
\begin{eqnarray}\label{HellingerExpo}
H^{(E)}(N,\mathbb H):={1\over{2}}\langle N^c\rangle^{\mathbb H}+\sum\left((1+\Delta N)\ln(1+\Delta N)-\Delta N\right).\end{eqnarray}
3) Let $Q^1$ and $Q^2$ be two probabilities such that $Q^2\ll Q^1$ and $T\in (0,+\infty)$. If $Q^i_T:=Q^i\big|_{{\cal H}_T}$ denote the restriction of $Q^i$ on ${\cal H}_T$  ($i=1,2$), then  
\begin{eqnarray}\label{entropy}
{\cal H}_{\mathbb H}(Q^1_T\big| Q^2_T):=E_{Q^2}\left[{{dQ^1_T}\over{dQ^2_T}}\ln\left({{dQ^1_T}\over{dQ^2_T}}\right)\right].\end{eqnarray}
\end{definition}
Now, we are in the stage of stating one of our principal results of this section.
\begin{theorem}\label{LogOPexistence} Suppose $G>0$. Then the following assertions are equivalent.\\
{\rm{(a)}} Log-optimal portfolio for  $(S^{\tau},\mathbb G)$ exists,  and $\widetilde\varphi^{\mathbb G}$ denotes its portfolio rate.  \\
{\rm{(b)}} There exist $K\in{\cal M}_{loc}(\mathbb F)$ and a nondecreasing and $\mathbb F$-predictable process $V$  such that $K_0=V_0=0$,  ${\cal E}(K)\exp(-V)\in {\cal D}(S,\mathbb F)$, and 
\begin{eqnarray*}
E\left[\left({\widetilde G}\is (V+H^{(0)}(K,P))\right)_T+(G_{-}\is h^{(E)}(G_{-}^{-1}\is m,P))_T-\langle K,m\rangle^{\mathbb F}_T\right]<+\infty.\end{eqnarray*}
 {\rm{(c)}} The dual problem (\ref{dualproblem}) admits a unique solution, i.e. there exists a unique ${\widetilde Z}^{\mathbb G}\in {\cal D}_{log}(S^{\tau},\mathbb G)$ such that 
\begin{eqnarray*}\label{MinimizationLog}
\min_{Z\in {\cal D}(S^{\tau},\mathbb G)}E\Bigl[-\ln(Z_T)\Bigr]= E\Bigl[-\ln({\widetilde Z}^{\mathbb G}_T)\Bigr].  \end{eqnarray*}
 {\rm{(d)}} There exists ${\widetilde Z}^{\mathbb F}\in {\cal D}(S, \mathbb F)$ such that $({\widetilde Z}^{\mathbb F})^{\tau}/{\cal E}(G_{-}^{-1}\is m)^{\tau}\in {\cal D}_{log}(S^{\tau},\mathbb G)$ and 
\begin{eqnarray}\label{MinimiDualEquivLog}
\inf_{Z^{\mathbb G}\in {\cal D}(S^{\tau},\mathbb G)}E\Bigl[-\ln(Z^{\mathbb G}_T)\Bigr]=E\Bigl[-\ln({\widetilde Z}^{\mathbb F}_{T\wedge{\tau}}/{\cal E}_{\tau\wedge T}(G_{-}^{-1}\is m))\Bigr].
\end{eqnarray}

 Furthermore, when the triplet $( \widetilde\varphi^{\mathbb G}, {\widetilde Z}^{\mathbb G}, {\widetilde Z}^{\mathbb F})$ exists, it satisfies the following
 \begin{eqnarray}\label{Relationship}
 {\cal E}( \widetilde\varphi^{\mathbb G}\is S^{\tau})={1\over{{\widetilde Z}^{\mathbb G}}}={{{\cal E}(G_{-}^{-1}\is m)^{\tau}}\over{({\widetilde Z}^{\mathbb F})^{\tau}}}.\end{eqnarray}
\end{theorem}
The importance of this theorem resides in assertions (b) and (d) and the second equality in (\ref{Relationship}). Assertion (b) gives us a practical way to check the existence of log-optimal portfolio for  $(S^{\tau},\mathbb G)$, while assertion (d) and the second equality in (\ref{Relationship}) describe the structures of the optimal deflator dual to log-optimal portfolio. In fact, in virtue to \cite{ChoulliYansori1} --see also Theorem \ref{GeneralDeflators}--, any deflator $Z^{\mathbb G}$ for  $(S^{\tau},\mathbb G)$ is the product of three orthogonal positive local martingales that are deflators for three orthogonal risks respectively, and is represented by a triplet of $\mathbb F$-observable processes $(Z^{\mathbb F}, \varphi^{(0)}, \varphi^{(pr)})$. Assertion (d) claims that the optimal deflator does not bear ''pure default" risks at all, and its triplet characterization takes the form of $({\widetilde Z}^{\mathbb F}, 0,0)$.
\begin{proof}{\it Theorem \ref{LogOPexistence}}  The proof of (a)$\Longleftrightarrow$ (c) is a direct application  of Theorem \ref{LemmaCrucial} for the model $(X,\mathbb H)=(S^{\tau}, \mathbb G)$ (see the equivalence (c) $\Longleftrightarrow$ (d) of this latter theorem), while (d)$\Longrightarrow$ (c) is obvious.  Thus, the rest of the proof focuses on the remaining equivalences and implications. To this end, we start deriving some important and useful remarks that we summarize in the following lemma.
\begin{lemma}\label{Lemma4proof} The following assertions hold.\\
 {\rm{(i)}}  For any $Z^{\mathbb G}\in {\cal D}_{log}(S^{\tau},\mathbb G)$, there exists $Z^{\mathbb F}\in {\cal D}(S,\mathbb F)$ such that 
 \begin{eqnarray}\label{DominationPrincipal}
 E\left[ -\ln( Z^{\mathbb G}_T)\right]\geq E\Bigl[-\ln\left({{Z^{\mathbb F}_{T\wedge\tau}/{\cal E}_{T\wedge\tau}(G_{-}^{-1}\is m)}}\right)\Bigr].\end{eqnarray}
  {\rm{(ii)}} We always have 
   \begin{eqnarray}\label{MinimizationReduction}
 \inf_{ Z^{\mathbb G}\in{\cal D}(S^{\tau},\mathbb G)} E\left[ -\ln( Z^{\mathbb G}_T)\right]= \inf_{ Z^{\mathbb F}\in{\cal D}(S,\mathbb F)}E\Bigl[-\ln\left({{Z^{\mathbb F}_{T\wedge\tau}/{\cal E}_{T\wedge\tau}(G_{-}^{-1}\is m)}}\right)\Bigr].\end{eqnarray}
 {\rm{(iii)}}  If $Z^{\mathbb F}\in {\cal D}(S,\mathbb F)$ such that $(Z^{\mathbb F})^{\tau}/{\cal E}(G_{-}^{-1}\is m)^{\tau}\in {\cal D}_{log}(S^{\tau},\mathbb G)$,  there exist $K\in {\cal M}_{0,loc}(\mathbb F)$ and a nondecreasing and $\mathbb F$-predictable process $V$ such that $V_0=K_0=0$, $Z^{\mathbb F}:={\cal E}(K)\exp(-V)$,  and 
  \begin{eqnarray}\label{Equality100}
&&E\left[-\ln\left({{Z^{\mathbb F}_{T\wedge\tau}}\over{{\cal E}_{T\wedge\tau}(G_{-}^{-1}\is m)}}\right)\right]=\\
&&E\left[(G_{-}\is V+{\widetilde G}\is H^{(0)}(K,P))_T-\langle K, m\rangle^{\mathbb F}_T+G_{-}\is h^{E}(G_{-}^{-1}\is m,P)_T\right].\nonumber\end{eqnarray}
\end{lemma}
Then the proof of (c)$\Longrightarrow$ (d) follows from combining (\ref {MinimizationReduction}) with assertion (i) of the lemma applied to ${\widetilde Z}^{\mathbb G}$. 
Therefore, the proof of the theorem will end as soon we prove  (c)$\Longleftrightarrow$ (b).  To this end, due to Theorem \ref{LemmaCrucial}, assertion (c) is equivalent to ${\cal D}_{log}(S^{\tau},\mathbb G)\not=\emptyset$. Thanks to (\ref{MinimizationReduction}) again, this latter claim holds if and only if there exists $Z^{\mathbb F}\in {\cal D}(S, \mathbb F)$ such that $(Z^{\mathbb F})^{\tau}/{\cal E}(G_{-}^{-1}\is m)^{\tau}\in {\cal D}_{log}(S^{\tau},\mathbb G)$. Thus, the equivalence  (c)$\Longleftrightarrow$ (b) follows immediately form combining these facts with Lemma \ref{Lemma4proof}-(iii). Therefore, the rest of this proof focuses on proving Lemma \ref{Lemma4proof} in two parts.\\
{\bf Part 1.} Here we prove assertions (i) and (ii) of Lemma \ref{Lemma4proof}.
On the one hand, for any $Z^{\mathbb G}\in{\cal D}(S^{\tau},\mathbb G)$, we apply Theorem \ref{GeneralDeflators} and deduce the existence of a triplet $\left(Z^{\mathbb F}, \varphi^{(o)}, \varphi^{(pr)}\right)$ that belongs to  ${\cal D}(S, \mathbb F)\times {\cal I}^o_{loc}(N^{\mathbb G},\mathbb G)\times L^1_{loc}({\rm{Prog}}(\mathbb F),P\otimes D)$ and satisfies  
$$
\varphi^{(pr)}>-1,\quad P\otimes D-a.e.,\quad 
-{{\widetilde G}\over{ G}}<\varphi^{(o)}<{{\widetilde G}\over{\widetilde G -G}},\quad\quad\quad P\otimes D^{o,\mathbb F}\mbox{-a.e..} $$
and  
\begin{eqnarray}\label{ZG2ZF}
Z^{\mathbb G}={{(Z^{\mathbb F})^{\tau}}\over{{\cal E}(G_{-}^{-1}\is m)^{\tau}}}{\cal E}(\varphi^{(0)}\is N^{\mathbb G}){\cal E}(\varphi^{(pr)}\is D).\end{eqnarray}
This implies that
 $$
 -\ln( Z^{\mathbb G})=-\ln({{(Z^{\mathbb F})^{\tau}/{\cal E}(G_{-}^{-1}\is m)^{\tau}}})-\ln( {\cal E}(\varphi^{(0)}\is N^{\mathbb G}))-\ln({\cal E}(\varphi^{(pr)}\is D)).$$
 Hence, thanks to Proposition \ref{Hzero4Log(Z)}-(c), we deduce that $Z^{\mathbb G}\in {\cal D}_{log}(S^{\tau},\mathbb G)$ if and only if $(Z^{\mathbb F})^{\tau}/{\cal E}(G_{-}^{-1}\is m)^{\tau}$ belongs to ${\cal D}_{log}(S^{\tau},\mathbb G)$ and $-\ln\left({\cal E}(\varphi^{(0)}\is N^{\mathbb G})\right)$ and $-\ln\left({\cal E}(\varphi^{(pr)}\is D)\right)$ are uniformly integrable $\mathbb G$-submartingales, and (\ref{DominationPrincipal}) follows immediately. This proves assertion (a) of the lemma.
 On the other hand, thanks to Theorem \ref{GeneralDeflators} -(c), the process $Z^{\tau}/{\cal E}(G_{-}^{-1}\is m)^{\tau}$ always belongs to ${\cal D}(S^{\tau},\mathbb G)$ as soon as  $Z\in {\cal D}(S,\mathbb F)$, and hence
 $$\inf_{Z^{\mathbb G}\in {\cal D}(S^{\tau},\mathbb G)}E\Bigl[-\ln(Z^{\mathbb G}_T)\Bigr]\leq \inf_{Z\in {\cal D}(S,\mathbb F)}E\Bigl[-\ln(Z_{T\wedge{\tau}}/{\cal E}_{\tau\wedge T}(G_{-}^{-1}\is m))\Bigr].$$
 Therefore, by combining  this latter inequality with assertion (i) of the lemma, we conclude that (\ref{MinimizationReduction}) always holds, and assertion (ii) is proved. \\
 {\bf Part 2.} Here, we prove assertion (iii) of the lemma. To this end, we consider $Z^{\mathbb F}\in  {\cal D}(S,\mathbb F)$. Thus, thanks to Proposition \ref{Hzero4Log(Z)}-(a), we deduce the existence of $K\in {\cal M}_{0,loc}(\mathbb F)$ and a nondecreasing and $\mathbb F$-predictable process $V$ such that $V_0=K_0=0$ and $Z^{\mathbb F}:={\cal E}(K)\exp(-V)$. Therefore, we derive 
 \begin{eqnarray}\label{Optimization1}
&&-\ln\left({{(Z^{\mathbb F})^{\tau}}\over{{\cal E}(G_{-}^{-1}\is m)^{\tau}}}\right)=-\ln((Z^{\mathbb F})^{\tau})+\ln( {\cal E}(G_{-}^{-1}\is m)^{\tau})\nonumber\\
&&={\mathbb G}\mbox{-local martingale}-{{I_{\Lbrack0,\tau\Lbrack}}\over{G_{-}}}\is\langle K, m\rangle^{\mathbb F}+H^{(0)}(K,P)^{\tau}\nonumber\\
&&+V^{\tau}+{{I_{\Lbrack0,\tau\Lbrack}}\over{G_{-}^2}}\is\langle m\rangle^{\mathbb F}-H^{(0)}(G_{-}^{-1}\is m,P)^{\tau}.\end{eqnarray}
Remark that the two processes $I_{\Rbrack0,\tau\Lbrack}G_{-}^{-2}\is\langle m\rangle^{\mathbb F}$ and $H^{(0)}(G_{-}^{-1}\is m,P)^{\tau}$ have variations that are $\mathbb F$-locally integrable and 
\begin{eqnarray}\label{HellingerE}
&&\left({{I_{\Lbrack0,\tau\Lbrack}}\over{G_{-}^2}}\is\langle m\rangle^{\mathbb F}-H^{(0)}(G_{-}^{-1}\is m,P)^{\tau}\right)^{p,\mathbb F}\nonumber\\
&&={1\over{G_{-}}}\is\langle m\rangle^{\mathbb F}-\left({\widetilde G}\is H^{(0)}({1\over{G_{-}}}\is m,P)\right)^{p,\mathbb F}\nonumber\\
&&={1\over{2G_{-}}}\is\langle m^c\rangle^{\mathbb F}+G_{-}\is \left(\sum ({{\Delta m}\over{G_{-}}})^2\right)^{p,\mathbb F}-\left(\sum {\widetilde G}({{\Delta m}\over{G_{-}}}-\ln(1+{{\Delta m}\over{G_{-}}}))\right)^{p,\mathbb F}\nonumber\\
&&={1\over{2G_{-}}}\is\langle m^c\rangle^{\mathbb F}+\left(\sum G_{-}\left((1+{{\Delta m}\over{G_{-}}})\ln(1+{{\Delta m}\over{G_{-}}}) -{{\Delta m}\over{G_{-}}}\right)\right)^{p,\mathbb F}\nonumber\\
&&=G_{-}\is h^{(E)}(G_{-}^{-1}\is m,P).
\end{eqnarray}
Thus, by combining this with (\ref {Optimization1}), Proposition \ref{Hzero4Log(Z)}-(b), and the easy fact that $U\in {\cal A}^+(\mathbb G)$ iff $U^{p,\mathbb F}\in {\cal A}^+(\mathbb F)$ for any nondecreasing process $U$, the equality (\ref{Equality100}) follows immediately. This ends the proof of the lemma.
\qed
\end{proof}
At the practical level, one starts with an initial model $(S,\mathbb F)$ admitting log-optimal portfolio, and tries to describe models of $\tau$ allowing $(S^{\tau},\mathbb G)$ to admit log-optimal portfolio. This boils down to the question (\ref{Q3}), and the answer to this follows from Theorem \ref{LogOPexistence} and is given by the following. 
\begin{theorem}\label{thoeremAppl1} Suppose $G>0$, and log-optimal portfolio for $(S,\mathbb F)$ exists. Then log-optimal portfolio for $(S^{\tau},\mathbb G)$ exists  if and only if
\begin{eqnarray}\label{EntropyCondition}
{\cal H}_{\mathbb G}(P_T\big| {\widetilde Q}_T)=E\left[ (G_{-}\is h^{(E)}(G_{-}^{-1}\is m,P))_T\right]<+\infty.
\end{eqnarray}
Here $h^{(E)}(N,P)$, for any $N\in {\cal M}_{0,loc}(\mathbb  F)$ with $1+\Delta N\geq 0$, and  the entropy ${\cal H}_{\mathbb G}(P_T\big| {\widetilde Q}_T)$ are given by Definition \ref{Hellinger}, and ${\widetilde Q}_T$ is defined in (\ref{Qtilde}).
\end{theorem}

\begin{proof} Due to Theorem \ref{LemmaCrucial} and Proposition \ref{Hzero4Log(Z)}-(b),  $(S,\mathbb F)$ admits log-optimal portfolio iff there exist $K\in {\cal M}_{loc}(\mathbb F)$ and a nondecreasing and $\mathbb F$-predictable process $V$ such that $K_0=V_0=0$, $Z:={\cal E}(K)\exp(-V)\in {\cal D}(S,\mathbb F)$, and 
$$E[-\ln(Z_T)]=E[V_T+H^{(0)}(K,P)_T]<+\infty.$$
Thus, due to Lemma \ref{H0toH1martingales}, we conclude that $\sqrt{[K,K]}$ is an integrable process ( or equivalently $K$ is a martingale such that $\displaystyle\sup_{0\leq t\leq T}\vert K_t\vert\in L^1(P)$), and hence the process $\langle K,m\rangle^{\mathbb F}$ has integrable variation as $m$ is a BMO $\mathbb F$-martingale. I.e. there exists a positive constant $C>0$ such that for any $\mathbb F$-stopping times $\sigma_1$ and $\sigma_2$ such that $\sigma_1\leq \sigma_2$, we have 
\begin{eqnarray*}
(\Delta m_{\sigma_1})^2+ E\left[[m,m]_{\sigma_2}-[m,m]_{\sigma_1}\big|{\cal F}_{\sigma_1}\right]\leq C\quad P\mbox{-a.s.}.\end{eqnarray*} Therefore, the process 
$G_{-}\is V+({\widetilde G}\is H^{(0)}(K,P))^{p,\mathbb F}-\langle K,m\rangle^{\mathbb F}$ belongs to ${\cal A}(\mathbb F)$. Thus, in virtue of Theorem  \ref{LogOPexistence}, log-optimal portfolio for $(S^{\tau},\mathbb G)$ exists if and only if the second condition in (\ref{EntropyCondition}) holds. Thus, the proof of the theorem follows immediately from the following two equalities
\begin{eqnarray}\label{equa401} {\cal H}_{\mathbb G}\left(P\big|{\widetilde Q}_T\right)&&=E\left[\ln({\cal E}_{T\wedge\tau}( {1\over{G_{-}}}\is m))\right]=E\left[G_{-}\is h^{(E)}({1\over{G_{-}}}\is m, \mathbb F)_T\right].\hskip 0.85cm
\end{eqnarray}
The last equality is a direct result of (\ref{HellingerE}). This proves the theorem.\qed
\end{proof}
In the remaining part of this section,  we {\it naturally} connects log-optimal portfolio of $(S^{\tau},\mathbb G)$  with log-optimal portfolio of $(S,\mathbb F)$ under adequate change of probability, and with the optimal portfolio for the economic model $(S,\mathbb F, \widetilde{U})$ where $\widetilde U$ is a random field utility that will be specified. To this end, we define the set of admissible portfolios for $\widetilde U$ as follows
\begin{eqnarray}\label{AdmissibleSet4Utilde}
\hskip -0.70cm{\cal A}_{adm}(S, \widetilde U):=\left\{\theta\in L(S, \mathbb F)\Bigg|\begin{array}{lll}1+\theta\is S>0,\\
E\left[\max(0,-{\widetilde U}(T, 1+(\theta\is S)_T))\right]<+\infty\end{array}\right\}.
\end{eqnarray}
\begin{theorem}\label{theorem3.8}Suppose $G>0$ and
\begin{eqnarray}\label{FiniteEntropy}
E\Bigl[{\cal E}_T(G_{-}^{-1}\is m)\ln\left({\cal E}_T(G_{-}^{-1}\is m)\right)\Bigr]<+\infty.\end{eqnarray} 
Then the following assertions are  equivalent  and sufficient for the existence of log-optimal portfolio for the model $(S^{\tau},\mathbb G)$.\\
{\rm{(a)}}  $({\widetilde U}, S, \mathbb F)$ admits an optimal portfolio, with ${\widetilde U}(t,x):={\cal E}_t(G_{-}^{-1}\is m)\ln(x)$ for $ x>0$, i.e. there exists $\theta^*\in {\cal A}_{adm}(S, \widetilde U)$ such that 
\begin{eqnarray}\label{Utilde}
\max_{\theta\in{\cal A}_{adm}(S,\widetilde U)} E\left[{\widetilde U}(T, 1+(\theta\is S)_T)\right]= E\left[{\widetilde U}(T, 1+(\theta^*\is S)_T)\right].\end{eqnarray}
 {\rm{(b)}} Log-optimal portfolio for $(S, {\widehat Q}, \mathbb F)$ exists, where $d{\widehat Q}:={\cal E}_T(G_{-}^{-1}\is m)dP$.\\
Furthermore, the three portfolios coincide on $\Lbrack0,\tau\Lbrack$ when they exists.
\end{theorem}

\begin{proof} Remark that, by combining \cite[Proposition 3.6]{ChoulliStricker2005} or \cite[Theorem 2.9]{ChoulliStricker2006} and $G_{-}=G_0{\cal E}_{-}(-{\widetilde G}^{-1}\is D^{o,\mathbb F}){\cal E}_{-}(G_{-}^{-1}\is m)\leq {\cal E}_{-}(G_{-}^{-1}\is m)$, we derive
\begin{eqnarray*}
E\Bigl[{\cal E}_T(G_{-}^{-1}\is m)\ln\left({\cal E}_T(G_{-}^{-1}\is m)\right)\Bigr]&&=E\left[\int_0^T{\cal E}_{s-}(G_{-}^{-1}\is m) dh^{(E)}_s(G_{-}^{-1}\is m, P)\right]\\
&&\geq E\left[\left(G_{-}\is h^{(E)}(G_{-}^{-1}\is m, P)\right)_T\right].
\end{eqnarray*}
Thus, if (\ref{FiniteEntropy}) holds, then the condition (\ref{EntropyCondition}) holds, ${\cal E}(G_{-}^{-1}\is m)^T$ becomes a uniformly integrable martingale, and $d{\widehat Q}:={\cal E}_T(G_{-}^{-1}\is m)dP$ is a well defined probability measure with finite entropy.  Therefore, on the one hand, direct calculations show that  ${\cal A}_{adm}(S,\widetilde U)=\Theta(S, \mathbb F, {\widehat Q})$. On the other hand, for any $\theta\in{\cal A}_{adm}(S,\widetilde U)$, 
\begin{eqnarray*}
E\left[{\widetilde U}(T, 1+(\theta\is S)_T)\right]=E^{{\widehat Q}}\left[\ln( 1+(\theta\is S)_T)\right].\end{eqnarray*}
This proves the equivalence between assertions (a) and (b). The rest of this proof focuses on proving that these assertions imply the existence of log-optimal portfolio for $(S^{\tau}, \mathbb G)$. To this end, we suppose that log-optimal portfolio of $(S, \mathbb F, {\widehat Q})$ exists. Thanks to Theorem \ref{LemmaCrucial}, this fact is equivalent to the existence of $Z:={\cal E}(K)e^{-V}\in {\cal D}(S, \mathbb F)$, where $K\in {\cal M}_{loc}(\mathbb F)$ and $V$ is a nondecreasing and $\mathbb F$-predictable such that $Z/ {\cal E}(G_{-}^{-1}\is m)\in {\cal D}_{log}(S,\mathbb F, {\widehat Q})$. Due to (\ref{FiniteEntropy}), $K-\langle K, G_{-}^{-1}\is m\rangle^{\mathbb F}\in {\cal M}_{loc}(\mathbb F, {\widehat Q})$ and direct Ito's calculations, we derive
\begin{eqnarray*}
&&E^{\widehat Q}\left[-\ln\left(Z_T/ {\cal E}_T(G_{-}^{-1}\is m)\right)\right]\\
&&= E^{\widehat Q}\left[-\ln\left(Z_T\right)\right]+E^{\widehat Q}\left[\ln\left( {\cal E}_T(G_{-}^{-1}\is m)\right)\right]\\
&&=G_0^{-1}E\left[\int_0^T {\cal E}_{s-}(G^{-1}\is D^{o,\mathbb F})^{-1}d U_s\right]+E\left[{\cal E}_T(G_{-}^{-1}\is m)\ln\left({\cal E}_T(G_{-}^{-1}\is m)\right)\right],\\
&&\geq G_0^{-1}E[U_T] +E\Bigl[{\cal E}_T(G_{-}^{-1}\is m)\ln\left({\cal E}_T(G_{-}^{-1}\is m)\right)\Bigr],
\end{eqnarray*}
where $U:=G_{-}\is V+{\widetilde G}\is H^{(0)}(K, P)-\langle K, m\rangle^{\mathbb F}$. This proves that, under (\ref{FiniteEntropy}), the condition $E^{\widehat Q}\left[-\ln\left(Z_T/ {\cal E}_T(G_{-}^{-1}\is m)\right)\right]<+\infty$ implies that $E[U_T]<+\infty$, which is equivalent to the existence of log-optimal portfolio for $(S^{\tau}, \mathbb G)$ in virtue of Theorem \ref{LogOPexistence}. This ends the proof of the theorem.\qed
\end{proof}
\begin{corollary}\label{pseudo}
Suppose $G>0$, and log-optimal portfolio for $(S,\mathbb F)$ exists.  Then the following conditions are all sufficient for the existence of log-optimal portfolio of $(S^{\tau},\mathbb G)$. Furthermore, both portfolios coincide on $\Lbrack0,\tau\Lbrack$.\\
{\rm{(a)}} $\tau$ is a pseudo-stopping time.\\
{\rm{(b)}} $\tau$ is independent of $\mathbb F$.\\
{\rm{(c)}} Every $\mathbb F$-martingale is a $\mathbb G$-local martingale (i.e. immersion holds).
\end{corollary}

\begin{proof}  On the one hand, it is easy to remark that assertion (a) is implied by either assertions (b)  or (c).  On the other hand, when assertion (a) holds, the existence of log-optimal portfolio for $(S^{\tau},\mathbb G)$ is a direct application of Theorem \ref{theorem3.8} and the fact that $m\equiv m_0$ --as this is a characterization for $\tau$ being a pseudo-stopping time due to \cite{Nik2005}--, which implies (\ref{FiniteEntropy}). This proves the corollary. \qed
\end{proof}
\section{ Log-optimal portfolio for $(S^{\tau},\mathbb G)$: Description and applications}\label{section5}
This section addresses the explicit computation of num\'eraire and log-optimal portfolios for $(S^{\tau},\mathbb G)$ in terms of $\mathbb F$-observable data of the model, and this will answer completely (\ref{Q4}). In virtue of \cite{ChoulliYansori2}, see also Theorem \ref{LemmaCrucial}, this will be possible only by using the {\it predictable characteristics} of semimartingales which are technical but powerful statistical tools for parametrization. Thus, this section starts recalling these characteristics for a general model $( X,\mathbb H)$ defined on the complete probability space $(\Omega, {\cal F}, P)$, where $\mathbb H$ is a filtration satisfying the usual conditions of completeness and right continuity, and $X$ is an $\mathbb H$-semimartingale.  We denote
$$\label{sigmaFields}
\widetilde {\cal O}(\mathbb H):={\cal O}(\mathbb H)\otimes {\cal
B}({\mathbb R}^d),\ \ \ \ \ \widetilde{\cal P}(\mathbb H):= {\cal
P}(\mathbb H)\otimes {\cal B}({\mathbb R}^d),
$$
where ${\cal B}({\mathbb R}^d)$ is the Borel $\sigma$-field on
${\mathbb R}^d$, the $\mathbb H$-optional and $\mathbb H$-predictable $\sigma$-fields respectively on the $\Omega\times[0,+\infty)\times{\mathbb R}^d$.  With a c\`adl\`ag $\mathbb H$-adapted process
$X$, we associate the optional random measure $\mu_X$
defined by
\begin{eqnarray*}\label{mesuresauts}
\mu_X(dt,dx):=\sum_{u>0} I_{\{\Delta X_u \neq 0\}}\delta_{(u,\Delta
X_u)}(dt,dx)\,.\end{eqnarray*}
For a product-measurable functional
$W\geq 0$ on $\Omega\times \mathbb R_+\times{\mathbb R}^d$, we
denote $W\star\mu_X$ (or sometimes, with abuse of notation,
$W(x)\star\mu_X$) the process
\begin{eqnarray*}\label{Wstarmu}
(W\star\mu_X)_t:=\int_0^t \int_{{\mathbb R}^d\setminus\{0\}}
W(u,x)\mu_X(du,dx)=\sum_{0<u\leq t} W(u,\Delta X_u) I_{\{ \Delta
X_u\not=0\}}.\end{eqnarray*}
On
$\Omega\times\mathbb R_+\times{\mathbb R}^d$, we define the
measure $M^P_{\mu_X}:=P\otimes\mu_X$ by $$M^P_{\mu_X}\left(W\right):=\int W
dM^P_{\mu_X}:=E\left[(W\star\mu)_\infty\right],$$ (when the expectation is well defined). The \emph{conditional expectation} given $
\widetilde{\cal P}$ of a product-measurable functional
$W$, denoted by $M^P_{\mu_X}(W\big|\widetilde{\cal P})$, is the unique $ \widetilde{\cal P}$-measurable
functional $\widetilde W$ satisfying $
E\left[(W I_{\Sigma}\star\mu_X)_\infty \right]=E\left[({\widetilde W}
I_{\Sigma}\star\mu_X)_\infty \right]$ for any $\Sigma$ belonging to $\widetilde{\cal P}.$
For the reader's convenience, we recall {\it the
 canonical decomposition} of $X$ (for more related details, we refer the reader to \cite[Theorem 2.34, Section II.2]{JS03})
\begin{equation}\label{modelSbis}
X=X_0+X^c+h\star(\mu_X-\nu_X)+b\is A^X+(x-h)\star\mu_X,\end{equation}
where $h$, defined as $h(x):=xI_{\{ \vert x\vert\leq 1\}}$, is the
truncation function, and \mbox{$h\star(\mu_X-\nu_X)$} is the unique pure jump $\mathbb H$-local martingale with jumps given by $h(\Delta X)I_{\{\Delta X\not=0\}}-\ ^{p,\mathbb H}(h(\Delta X)I_{\{\Delta X\not=0\}})$. For the matrix $C^X$ with entries $C^{ij}:=\langle X^{c,i},
X^{c,j}\rangle $, and $\nu_X$, we can find a version satisfying
$$  C^X=c^X\is A^X,\ \nu_X(d t,\ d x)=d A_t^XF_t^X(d x),\
F_t^X(\{0\})=0,\ \displaystyle{\int} (\vert
x\vert^2\wedge 1)F_t^X(d x)\leq 1.
$$ Here $A^X$ is increasing and predictable, $b^X$
and $c^X$ are predictable processes,
 $F_t^X(d x)$ is a predictable kernel, $b_t^X(\omega)$ is a vector in $\hbox{I\kern-.18em\hbox{R}}^d$ and
$c_t^X(\omega)$ is a symmetric $d\times d$-matrix, for all $(\omega,\
t)\in\Omega\times \mathbb R_+$. For $W\geq 0$ and $\widetilde{\cal P}$-measurable, we put 
\begin{eqnarray}\label{What}
{\widehat W}_t:=\int W(t,x)\nu_X(\{t\}, dx),\quad a_t:=\widehat{1}_t=\nu_X(\{t\},\mathbb R^d).
\end{eqnarray}
The quadruplet
\begin{eqnarray}\label{PCharac4X}
(b^X,c^X, F^X, A^X)\ \mbox{are the predictable characteristics of}\ (X,\mathbb H).\end{eqnarray}

In the rest of paper, except the appendix, $(X, \mathbb H)\in\{(S, \mathbb F),(S^{\tau},\mathbb G)\}$. Thus, throughout the rest of the paper, for the sake of simplicity, the random measure $\mu_S$ associated with the jumps of $S$ will be denoted by $\mu$, while $S^c$ denotes the continuous $\mathbb F$-local martingale
part of $S$, and the quadruplet
\begin{eqnarray*}\label{FpredictCharac}
\left(b,c,F, A\right)\ \mbox{are the predictable characteristics of}\ (S,\mathbb F).\end{eqnarray*}
Or equivalently {\it the
 canonical decomposition} of $S$ (see Theorem 2.34, Section II.2 of \cite{JS03} for details) is given by 
\begin{equation}\label{modelSbis}
S=S_0+S^c+h\star(\mu-\nu)+b\is A+(x-h)\star\mu,\quad h(x):=xI_{\{ \vert x\vert\leq 1\}}.\end{equation}
Throughout the rest of this section, we consider Jacod's decomposition for the $\mathbb F$-martingale $G_{-}^{-1}\is m$, see Theorem \ref{tmgviacharacteristics} for details, \begin{eqnarray}\label{Jacod4m}
G_{-}^{-1}\is {m}&&= \beta \is S^c +U^{(m)} \star (\mu - \nu) + g_m\star \mu + m^{\perp},\\
&& U^{(m)}:=f_m-1+{{\widehat{f_m}-a}\over{1-a}}I_{\{a<1\}}.\nonumber\end{eqnarray}
We end this subsection by defining the space ${\cal L}(S,\mathbb F)$ and the function ${\cal K}_{log}$, that will be used throughout the paper, as follows.
\begin{eqnarray}
{\cal L}(S,\mathbb F)&&:=\Bigl\{\theta\in{\cal P}(\mathbb F)\ \big|\ 1+x^{tr}\theta_t(\omega)>0\quad P\otimes dA\otimes F\mbox{-.a.e.}\Bigr\},\label{SelL(F)}\\
{\cal K}_{log}(y)&&:={{ -y}\over{1 +y }}  
+\ln(1 +y)\quad \mbox{for any}\quad y>-1.\label{Kfunction}
\end{eqnarray}
 The rest of this section is divided into three subsections. The first subsection states the main results, while their proofs will detailed in the second subsection. The last subsection illustrates the main results on the case when the initial model $(S,\mathbb F)$ follows a jump-diffusion model.
\subsection{Main results and their applications and financial interpretations}\label{Subsection4Results}
This subsection describes explicitly log-optimal portfolio for $(S^{\tau},\mathbb G)$ using the model's data seen through $\mathbb F$ only, and discusses their applications afterwards.
\begin{theorem}\label{optimalportfoliogen} Suppose $G>0$, and let ${\cal K}_{log}$ be the function given by (\ref{Kfunction}). Then the following three assertions are equivalent.\\
{\rm{(a)}} Log-optimal portfolio for $(S^{\tau}, \mathbb G)$, denoted by ${\widetilde\theta}^{\mathbb G}$, exists. \\
{\rm{(b)}} There exists $\widetilde\varphi\in {\cal L}(S,\mathbb F)$ such that, for any $\theta\in{\cal L}(S,\mathbb F)$, the following hold
\begin{eqnarray}
&&\hskip -0.55cm(\theta-\widetilde\varphi)^{tr}(b +c(\beta-\widetilde\varphi)) + \int(\theta-\widetilde\varphi)^{tr}\left({{f_m(x) }\over{1 + \widetilde\varphi^{tr} x }} x-h(x)\right) F(dx)\leq 0,\label{Cond4ptimalityG}\\
 &&\hskip -0.55cmE\left[(G_{-}\is \widetilde V)_T+(G_{-}\widetilde\varphi^{tr}c\widetilde\varphi\is A)_T+(G_{-}{\cal K}_{log}(\widetilde\varphi^{tr}x) f^{(m)}\star\nu)_T\right]<+\infty,\label{Cond4integrabilityG}\\
  &&\hskip -0.55cm\widetilde{V}:=\left(\widetilde\varphi^{tr} b +\widetilde\varphi^{tr}c(\beta-\widetilde\varphi)\right)\is A  + \left({{f_m(x)\widetilde\varphi^{tr}x}\over{1 + \widetilde\varphi^{tr} x }}  -\widetilde\varphi^{tr}h(x)\right)\star\nu.\label{VtildeG} \end{eqnarray}
 {\rm{(c)}} Num\'eraire portfolio for $(S^{\tau},\mathbb G)$ exists, and its rate $\widetilde\varphi^{\mathbb G}$ satisfies (\ref{Cond4integrabilityG}). \\
  Furthermore, the processes ${\widetilde\theta}^{\mathbb G}$,  $\widetilde{\varphi}$,  $\widetilde\varphi^{\mathbb G}$ , ${\widetilde Z}^{\mathbb G}$ solution to (\ref{dualproblem}), and ${\widetilde Z}^{\mathbb F}\in {\cal D}(S,\mathbb F)$ described via (\ref{MinimiDualEquivLog}) are related to each other by the following  
\begin{eqnarray}
&&{\widetilde\theta}^{\mathbb G}\left(1+({\widetilde\theta}^{\mathbb G}\is S^{\tau})_{-}\right)^{-1}=\widetilde{\varphi}= \widetilde\varphi^{\mathbb G}\quad \mbox{on}\quad \Lbrack0,\tau\Lbrack,\label{LogOPrate4G}\\
&&\hskip -0.45cm {1\over{{\cal E}(\widetilde\varphi\is S)^{\tau}}}={\widetilde Z}^{\mathbb G}:={\cal E}({\widetilde K}^{\mathbb G}){\cal E}(-\widetilde{V}^{\tau})={{{\cal E}({\widetilde K}^{\mathbb F})^{\tau}{\cal E}(-\widetilde{V}^{\tau})}\over{{\cal E}(G_{-}^{-1}\is m)^{\tau}}}=:{{({\widetilde Z}^{\mathbb F})^{\tau}}\over{{\cal E}(G_{-}^{-1}\is m)^{\tau}}} ,\label{OptimalDualLoggeneral}\\
 &&{\widetilde K}^{\mathbb F}:=-\widetilde\varphi\is S^c+ {{\widetilde\Gamma}\over{G_{-}}}\is m -{{{\widetilde\Gamma}(\widetilde\varphi^{tr}x) f_m}\over{1+\widetilde\varphi^{tr}x}}\star(\mu-\nu)-{{{\widetilde\Gamma}(\widetilde\varphi^{tr}x)g_m}\over{1+\widetilde\varphi^{tr}x}}\star\mu.\hskip 1cm\label{K-F}\\
 &&{\widetilde K}^{\mathbb G}:=-{\widetilde\varphi}\is {\cal T}(S^c)+{{-{\widetilde\Gamma}{\widetilde\varphi}^{tr}x}\over{1+{\widetilde\varphi}^{tr}x}}I_{\Lbrack0,\tau\Lbrack}\star(\mu-f_m\cdot\nu)\label{K-G} \\
  &&\hskip -0.55cm\widetilde\Gamma:=\left(1-\widehat{f_m}+\widehat{f_m f^{(op)}}\right)^{-1},\quad f^{(op)}(x):=\left(1+\widetilde\varphi^{tr}x\right)^{-1}.\label{GammaTilde}\end{eqnarray}
\end{theorem}
Thanks to \cite[Theorems 2.17 and 2.20]{ChoulliDavelooseVanmaele}, the uncertainty in $\tau$ bears three types of orthogonal risks. These risks are the correlation risks whose generator is the $\mathbb F$-martingale $m$ (or equivalently ${\cal E}(G_{-}^{-1}\is m)$), the pure default risk of type one generated by $N^{\mathbb G}$ defined in (\ref{processNG}), and the second type of pure default risk that takes the form of $k\is D$ where $k$ spans the space $L^1_{loc}(\widetilde\Omega, \mbox{Prog}(\mathbb F), P\otimes D)$ and satisfies $E[k_{\tau}I_{\{\tau<+\infty\}}\big|{\cal F}_{\tau}]=0$ $P$-a.s..\\
In virtue of Theorems \ref{NumeraireGeneral} and \ref{LogOPexistence}, it is clear that the correlation risk is the only risk that impacts num\'eraire and log-optimal portfolios. Theorem 5.1 gives a more ``chirurgical" and deep description of the impact of $\tau$ on the log-optimal portfolio. To be more precise in our analysis, we appeal to (\ref{Jacod4m}), and write
\begin{eqnarray}\label{Decomposition4m}
{\cal E}(G_{-}^{-1}\is m)=\underbrace{{\cal E}\Biggl(\beta\is S^c+U^{(m)}\star(\mu-\nu)\Biggr)}_{\mbox{$S$-correlation source}}\times\underbrace{{\cal E}\Biggl( (g_m/f_m)\star\mu+m^{\perp}\Biggr)}_{\mbox{$S$-non-correlation source}}.\hskip 0.5cm
  \end{eqnarray}
  This decomposes the source of correlation risks between $\tau$ and $\mathbb F$ into the product of $S$-correlation and $S$-non-correlation risk generators, which leads naturally to the following definition.
  \begin{definition} Let $X$ be an $\mathbb F$-semimartingale.\\
  $\tau$ is ``{\it non-correlated}" to $X$ if $[m, M^X]\in {\cal M}_{loc}(\mathbb F)$ for any $M^X\in {\cal M}_{loc}(\mathbb F)$ generated by $X$. I.e.  it is of the form of $M^X=\varphi\is X^c+W\star(\mu^X-\nu^X)$, for $\mathbb F$-predictable  $\varphi$ and $W$ fulfilling the integrability conditions of Theorem \ref{tmgviacharacteristics}.
  \end{definition}
This defines a large class of random times that covers the class of pseudo-stopping times, and it possesses similar ``nice" properties as them. In fact, it is easy to prove that $\tau$ is an $\mathbb F$-pseudo-stopping time if and only if it is non-correlated to any bounded $\mathbb F$-martingale. Furthermore, in virtue of (\ref{Jacod4m}), it is clear that $\tau$ is non-correlated to $S$ if and only if $(\beta,f_m)\equiv (0,1)$ $P\otimes A\otimes F$-a.e.. Thus, being non-correlated to $S$ is a much weaker assumption than being a pseudo-stopping time. Furthermore, our notion of ``non-correlation" in Definition \ref{Decomposition4m} allows us to go deeper than Corollaries \ref{ParticularCases} and \ref{pseudo} in describing a larger class of random times $\tau$ that do not affect num\'eraire and log-optimal portfolios when passing from $(S, \mathbb F)$ to $(S^{\tau},\mathbb G)$.  This is the aim of the following.
\begin{proposition}\label{Corollary5.3} Suppose ${\cal D}(S,\mathbb F)\not=\emptyset$, and let  $T\in (0,+\infty)$. Then the following assertions hold.\\
{\rm{(a)}} If $\tau$ is non-correlated to $S$, then num\'eraire porfolios for $(S^{\tau},\mathbb G)$ and $(S, \mathbb F)$ coincide on $\Rbrack0,\tau\Lbrack$.\\
{\rm{(b)}} Suppose that $\tau$ is non-correlated to $S$. If log-optimal portfolio for $(S,\mathbb F)$ exists, then log-optimal portfolio for $(S^{\tau\wedge T}, \mathbb G)$ does exist also, and both portfolios coincide on $\Rbrack0,T\wedge\tau\Lbrack$.\\
{\rm{(c)}} Suppose that $S$ is continuous and denote by $\widetilde\lambda$ the num\'eraire portfolio rate for $(S,\mathbb F)$. Then the following properties hold.\\ 
{\rm{(c.1)}}  Num\'eraire portfolio rate for $(S^{\tau},\mathbb G)$ --denoted by $\widetilde\varphi$-- exists and is given by $\widetilde\varphi=\widetilde\lambda+\beta$.\\
{\rm{(c.2)}} $\widetilde\varphi$ is log-optimal portfolio rate for $(S^{T\wedge\tau},\mathbb G)$ if and only if
$$E\left[\int_0^T G_{s-}{\widetilde\varphi}_s^{tr}c_s{\widetilde\varphi}_s dA_s\right]<+\infty.$$
{\rm{(c.3)}} If $E\left[\int_0^T{\widetilde\lambda}_s^{tr}c_s{\widetilde\lambda}_s dA_s\right]<+\infty$, then log-optimal portfolio rate for $(S^{T\wedge\tau},\mathbb G)$  exists if and only if 
$$E\left[\int_0^T {1\over{G_{s-}}}d\langle m^c, m^c\rangle_s\right]=E\left[\int_0^T G_{s-}{\beta}_s^{tr}c_s{\beta}_s dA_s\right]<+\infty.$$
\end{proposition}
We relegate the proof this corollary to Subsection \ref{Subsection4Proofs}, for the sake of simple exposition, while herein we focus on discussing the meaning of its results and its connection to the literature. Both assertions (a) and (b) explain that the only risk that affects the structures and/or the existence of num\'eraire and log-optimal portfolios is the risk generated by the correlation between $\tau$ and $S$. Assertion (c) considers the case of continuous $S$, which is considered in the majority of (or all) the insider information literature, see \cite{amendingerimkellerschweizer98,ADImkeller,AImkeller,JImkellerKN,GrorudPontier,pikovskykaratzas96,kohatsusulem06}. Thus, assertion (c) shows what will become in our progressive setting those results of the insider-information framework using {\it information drifts}.\\
The converse of assertion (b) is not true in general. In other words,  in some cases, $\tau$ might have a positive impact on $(S, \mathbb F)$ instead. Indeed, for model $(S, \mathbb F)$ admitting num\'eraire portfolio with rate $\widetilde\lambda$ which is {\it locally} a log-optimal portfolio but {\it not globally}, we can always find models for $\tau$ such that $(S^{T\wedge\tau}, \mathbb G)$ admits log-optimal portfolio, as soon as the probability space $(\Omega, {\cal F}, P)$ is rich enough. This claim can be also viewed as a complement result to Theorem \ref{thoeremAppl1}, as it addresses a case when log-optimal portfolio for $(S,\mathbb F)$ might fail to exist.
\begin{theorem}\label{ConverseTheorem} Suppose that $S$ is quasi-left-continuous (i.e. it does not jump on predictable stopping times)  and $(\Omega, {\cal F}, P)$ supports an exponentially distributed random variable  $\xi$  independent of ${\cal F}_{\infty}:=\sigma(\cup_{s\geq 0}{\cal F}_s)$ and $E[\xi]=1$.  \\
If $\left({\cal D}_{log}(S, \mathbb F)\right)_{loc}\not=\emptyset$, then there exists a random time $\tau$ having a positive Az\'ema supermartingale and $(S^{\tau}, \mathbb G)$ admits log-optimal portfolio. Furthermore, this portfolio coincides with num\'eraire portfolio of $(S,\mathbb F)$ on $\Lbrack0,\tau\Lbrack$. \end{theorem}
It is clear that when num\'eraire portfolio rate $\widetilde\lambda$ exists (see Theorem \ref{LemmaCrucial}), 
\begin {eqnarray}\label{Psi-tilde}
\widetilde\Psi:=\left({\widetilde\lambda}^{tr}b-{1\over{2}}{\widetilde\lambda}^{tr}c\widetilde\lambda\right)\is A +\left(\ln(1 + \widetilde\lambda^{tr} x) -{\widetilde\lambda}^{tr}h\right)\star\nu
\end{eqnarray}
is a well defined, predictable, and nondecreasing process with values in $[0,+\infty]$.
Due to Theorem \ref{LemmaCrucial}, the assertions below are equivalent to $\left({\cal D}_{log}(S, \mathbb F)\right)_{loc}\not=\emptyset$:\\
i) Log-optimal portfolio for $(S,\mathbb F)$ exists locally (see \cite{ChoulliDengMa} for this concept). This means that there exists a sequence of stopping times $(T_n)_n$ that increases to infinity such that  each $(S^{T_n},\mathbb F)$ admits log-optimal portfolio.\\
ii) Num\'eraire portfolio exists with rate $\widetilde\lambda$ and the nondecreasing process ${\widetilde\Psi}$ defined by (\ref{Psi-tilde}) has finite values (i.e. ${\widetilde\Psi}_t<+\infty$ $P$-a.s. for all $t\in[0,+\infty)$).\\
We think that the claim of Theorem \ref{ConverseTheorem} remains valid under the weaker assumption ${\cal D}(S, \mathbb F)\not=\emptyset$. I.e. we believe that the following {\bf conjecture} holds:\\ {\it If $(S,\mathbb F)$ admits num\'eraire portfolio and $(\Omega, {\cal F}, P)$ is rich enough as in Theorem \ref{ConverseTheorem}, then there exists $\tau$ such that $(S^{\tau},\mathbb G)$ admits log-optimal portfolio}. 
\begin{proof}{\it of Theorem \ref{ConverseTheorem}} 
Consider a nonnegative and $\mathbb F$-predictable process $\Phi$ such that $\Phi_0=0$ and the process $\mbox{ess}\sup_{0<s\leq \cdot}\Phi_s$ has finite values, and put  
\begin{eqnarray}
\tau:=\inf\left\{t\geq 0\ \Bigg|\quad{\Phi}_t\geq \xi\right\}\quad\mbox{and}\quad  {\widetilde\Phi}_t:=\mbox{ess}\sup_{0<s\leq t}\Phi_s.\label{tauDefinition}
\end{eqnarray}
Then thanks to the independence between ${\cal F}_{\infty}$ and $\xi$, we calculate the triplet  $(G, \widetilde G, G_{-})$ associated to $\tau$ as follows.
\begin{eqnarray*}
G_t&&:=P(\tau>t\big|{\cal F}_t)=P(\xi>{\widetilde\Phi}_t\big|{\cal F}_t)=\exp(-{\widetilde\Phi}_t)>0,\\
 {\widetilde G}_t&&:=P(\tau\geq t\big|{\cal F}_t)=P(\xi\geq {\widetilde\Phi}_{t-}\big|{\cal F}_t)=\exp(-{\widetilde\Phi}_{t-})=G_{t-}
.\end{eqnarray*}
This proves that, in our current case, we have  the immersion case as $G$ is nondecreasing process. Hence $m\equiv m_0$, and by applying Theorem \ref{optimalportfoliogen} to this case and letting $T$ goes to $+\infty$, we deduce that num\'eraire portfolio for $(S^{\tau}, \mathbb G)$ exists and its rate coincides with $\widetilde\lambda$ on $\Lbrack0,\tau\Lbrack$. Thus, in order to complete the proof of the theorem we need to prove that  $\widetilde\lambda$ fulfills (\ref{Cond4integrabilityG}). To this end, in virtue of the above discussion (paragraph after the theorem), we remark that ${\widetilde\Psi}\in {\cal A}^+_{loc}.$
Hence, by choosing $\Phi={\widetilde\Psi}$ we deduce that $\widetilde\Phi={\widetilde\Psi}$ and due to $f_m\equiv 1$, $\beta\equiv 0$ and the continuity of $A$ --which is implied by the quasi-left-continuity of $S$--, the process $\widetilde V$ of (\ref{VtildeG}) satisfies  $\widetilde V+{1\over{2}}\widetilde\lambda^{tr}c{\widetilde\lambda}\is A+{\cal K}_{log}(\widetilde\lambda^{tr}x)\star\nu={\widetilde\Psi}$ and  
\begin{eqnarray*}
E\left[\int_0^{\infty} G_{s-}d{\widetilde\Psi}_s\right]=E\left[\int_0^{\infty} {\cal E}_{s-}(-{\widetilde\Psi})d{\widetilde\Psi}_s\right]=E\left[1-{\cal E}_{\infty}(- {\widetilde\Psi})\right]\leq 1.
\end{eqnarray*}
This proves that $\widetilde\lambda I_{\Lbrack0.\tau\Lbrack}$ is the log-optimal portfolio rate for $(S^{\tau},\mathbb G)$ and the proof of theorem is completed.\qed
\end{proof}
The following theorem answers (\ref{Q5}) by evaluating $u_T(S^{\tau}, \mathbb G)-u_T(S, \mathbb F)$ and singling out its important parts that we analyze  afterwards. 

\begin{theorem}\label{Difference4u}Suppose $G>0$, the log-optimal portfolio rate $\widetilde\lambda$  for $(S, \mathbb F)$ exists, and (\ref{EntropyCondition}) holds.  Then there exists $\widetilde\varphi\in {\cal L}(S, \mathbb F)$ satisfying (\ref{Cond4ptimalityG}),
and
\begin{eqnarray}
&&\Delta_T(S,\tau,\mathbb F):=u_T(S^{\tau}, \mathbb G)-u_T(S, \mathbb F)\nonumber\\
&&=-\underbrace{E\left[-({\widetilde G}\is {\widetilde{\cal H}}(\mathbb G))_T+({\widetilde G}\is {\widetilde{\cal H}}(\mathbb F))_T+\langle {\widetilde K}^{\mathbb F}- {\widetilde L}^{\mathbb F},  m\rangle^{\mathbb F}_T\right]}_{\mbox{correlation-risk}}\label{Delta(S,tau)1}\hskip 1cm\\
\nonumber\\
&&-\underbrace{E\left[\left((1-\widetilde G)\is{\widetilde{\cal H}}(\mathbb F)\right)_T\right]}_{\mbox{cost-of-leaving-earlier}}-\underbrace{E\left[\langle {\widetilde L}^{\mathbb F},m\rangle^{\mathbb F}_T\right]}_{\mbox{NP($\mathbb F$)-correlation}}+\underbrace{ {\cal H}_{\mathbb G}\left(P\big|{\widetilde Q}_T\right)}_{\mbox{information-premium}},\nonumber\\
&&=-\underbrace{E\left[\left((1-\widetilde G)\is{\widetilde{\cal H}}(\mathbb F)\right)_T\right]}_{\mbox{cost-of-leaving-earlier}}-\underbrace{E\left[\langle {\widetilde L}^{\mathbb F},m\rangle^{\mathbb F}_T\right]}_{\mbox{NP($\mathbb F$)-correl.}}+\underbrace{E\left[\int_0^T{\widetilde {\cal R}}_tG_{t-}dA_t\right]}_{\mbox{num\'eraire-change-premium}}\label{Delta(S,tau)2}\hskip 0.25cm
\end{eqnarray}
Furthermore, the ``correlation-risk", the ``cost-of-leaving-earlier", the ``information-premium", and the ``num\'eraire-change-premium" are nonnegative quantities.\\  
Here, ${\widetilde Q}_T$, $\widetilde V$ and $ {\widetilde K}^{\mathbb F}$ are given by (\ref{Qtilde}), (\ref{VtildeG}) and (\ref{K-F}) respectively, and  the processes ${\widetilde {\cal R}}$, ${\widetilde{\cal H}}(\mathbb G)$ and ${\widetilde{\cal H}}(\mathbb F)$ are given by 
\begin{eqnarray}
{\widetilde {\cal R}}_t&&:=({\widetilde\varphi}_t-{\widetilde\lambda}_t)^{tr}b_t+{\widetilde\varphi}^{tr}_t c_t(\beta_t-{1\over{2}}{\widetilde\varphi}_t)-{\widetilde\lambda}^{tr}_t c_t(\beta_t-{1\over{2}}{\widetilde\lambda}_t) \label{NumeraireChanegPremuim}\\
&&+\int \left(f_m(t,x)\ln((1+{\widetilde\varphi}^{tr}_t x)(1+{\widetilde\lambda}^{tr}_t x)^{-1})-({\widetilde\varphi}_t-{\widetilde\lambda}_t)^{tr} h(x)\right)F_t(dx) \nonumber\\
{\widetilde{\cal H}}(\mathbb G)&&:=\widetilde  V+\sum(-\Delta{\widetilde V}-\ln(1-\Delta{\widetilde V}))+H^{(0)}({\widetilde K}^{\mathbb F},P),\label{Htilde(G)}\\
{\widetilde{\cal H}}(\mathbb F)&&:={\widetilde  V}^{\mathbb F}+\sum(-\Delta{\widetilde  V}^{\mathbb F}-\ln(1-\Delta{\widetilde  V}^{\mathbb F}))+H^{(0)}({\widetilde L}^{\mathbb F},P),\label{Htilde(F)}
\end{eqnarray}
where ${\widetilde L}^{\mathbb F}$  and ${\widetilde  V}^{\mathbb F}$  are defined by
\begin{eqnarray}
&&{\widetilde L}^{\mathbb F}:=-{\widetilde\lambda}\is S^c-{{{\widetilde\Xi}{\widetilde\lambda}^{tr}x}\over{1+{\widetilde\lambda}^{tr}x}}\star(\mu-\nu),\ {\widetilde\Xi}_t^{-1}:=1-a_t+\int{{\nu(\{t\},dx)}\over{1+{\widetilde\lambda_t}^{tr}x}},\hskip 1cm\label{Ltilde}\\
&&{\widetilde V}^{\mathbb F}:=\left({\widetilde\lambda}^{tr}b-{1\over{2}}{\widetilde\lambda}^{tr}c\widetilde\lambda\right)\is A +\left({{\widetilde\lambda^{tr} x}\over{1 + \widetilde\lambda^{tr} x}} -{\widetilde\lambda}^{tr}h\right)\star\nu.\label{Utilde}
\end{eqnarray}
\end{theorem}
Our theorem measures the quantity IEU$_{log}(S,\tau,\mathbb F)$, denoted by $\Delta_T(S,\tau,\mathbb F)$ and defined in (\ref{Delta(S, Tau)}), in different manners by quantifying deeply the various factors that directly affect log-optimal portfolio when passing from $(S, \mathbb F)$ to the stopped model $(S^{\tau}, \mathbb G)$. Two of these factors were {\it intuitively} known and understood, while herein we quantify them with a sharpe precision using the model's data that is observable in $\mathbb F$ no matter how general the model $(S, \mathbb F, \tau)$ is. In fact, we know intuitively that the agent endowed with the flow $\mathbb G$ has the information advantage of seeing the occurrence of $\tau$, and this {it naturally} should generate to her some premium that we quantified and called ``information-premium".  Simultaneously to this advantageous information, our $\mathbb G$-agent might face two types of risks due to the stochasticity of $\tau$. The first risk comes from the length of the random horizon, which might lead to the cost-of-leaving earlier, as the horizon $T\wedge\tau$ is shorter in general than T with positive probability. This fact is also well known for the case when $\tau$ is an $\mathbb F$-stopping time, a fortiori when it is fixed/deterministic horizon, due to the myopic feature of the logarithmic utility. The second risk, that our $\mathbb G$-investor might face, is essentially intrinsic to the correlation between $\tau$ and $S$.  We baptize this risk as ``correlation-risk"  and we precisely quantify it in terms of the  $\mathbb F$-observable model's data. This third factor appears naturally  in our analysis and was not intuitively clear, in contrast to the two previous factors. \\
The fourth factor, which appears in both expression (\ref{Delta(S,tau)1}) and (\ref{Delta(S,tau)2}), represents the correlation between $\tau$ and num\'eraire portfolio of $(S, \mathbb F)$. We naturally quantify this latter factor with the expression $E[\langle {\widetilde L}^{\mathbb F}, m\rangle^{\mathbb F}_T]$, as $ {\widetilde L}^{\mathbb F}$ is the main randomness-source of the optimal-deflator dual to num\'eraire portfolio.\\
By comparing (\ref{Delta(S,tau)1}) and (\ref{Delta(S,tau)2}), we obviously deduce that
\begin{eqnarray*}
\mbox{\it num\'eraire-change-premium}= \left(\mbox{\it information-premium}\right)-\left(\mbox{\it correlation-risk}\right)\geq 0.\end{eqnarray*}
This shows that the ``correlation-risk" liability does not exceed the ``information-premium", and hence it diminishes the informational advantage of our $\mathbb G$-investor without entailing a loss. This ``num\'eraire-change-premium" is null if and only if both num\'eraire portfolios $\widetilde\lambda$ and $\widetilde\varphi$ are equal in some sense.  This fact and further discussions on these factors are summarized in the following.\\
\begin{theorem}\label{Proposition4DifferenceU} Suppose that assumptions of Theorem \ref{Difference4u} hold. Then the following assertions hold.\\
{\rm{(a)}} The correlation-risk is null and both num\'eraire portfolios for $(S,\mathbb F)$ and $(S^{\tau},\mathbb G)$ coincide on $\Rbrack0,\tau\Lbrack$ if and only if the information-premium is null  if and only if $\tau$ is a pseudo-stopping time (i.e. $E[M_{\tau}]=E[M_0]$ for any bounded $\mathbb F$-martingale). In this case, the NP$(\mathbb F)$-correlation is also null, and 
\begin{eqnarray}\label{Case1}
\Delta_T(S,\tau,\mathbb F)=-\left(\mbox{cost-of-leaving-earlier}\right)\leq 0.\end{eqnarray}
{\rm{(b)}}  Suppose that $\tau$ is non-correlated to $S$. Then (\ref{Case1}) holds, the correlation-risk coincide with the information-premium, and the NP$(\mathbb F)$-correlation is null.
{\rm{(c)}} The ``num\'eraire-change-premium" is null iff the num\'eraire portfolio rates $\widetilde\lambda$ and  ${\widetilde\varphi}^{\mathbb G}$, for $(S, \mathbb F)$ and $(S^{\tau},\mathbb G)$ respectively, coincide on $\Lbrack0,\tau\Lbrack$, i.e. $\widetilde\lambda\is S^{\tau}={\widetilde\varphi}^{\mathbb G}\is S^{\tau}$ or on  $\Lbrack0,\tau\Lbrack$ $P\otimes A\otimes F$-a.e. $c\widetilde\lambda=c{\widetilde\varphi}^{\mathbb G}$, $b^{tr}\widetilde\lambda= b^{tr}{\widetilde\varphi}^{\mathbb G} $  and $x^{tr}\widetilde\lambda=x^{tr}{\widetilde\varphi}^{\mathbb G}$. \\
{\rm{(d)}} If $S$ is a local martingale, then 
\begin{eqnarray}\label{Case3}
\Delta_T(S,\tau,\mathbb F)=\mbox{num\'eraire-change-premium}\geq 0.\end{eqnarray}
{\rm{(e)}} Suppose $S$ is continuous. Then the pair $(\widetilde\lambda, \widetilde\varphi)$ defined in Theorem \ref{Difference4u} satisfies $\widetilde\lambda+\beta=\widetilde\varphi$ and the following hold.
\begin{eqnarray}
&&\mbox{information-premium}=E\left[\int_0^T{{G_{s-}}\over{2}}\beta^{tr}_s c_s\beta_s dA_s\right]+\mbox{correlation-risk},\hskip 0.75cm\label{Case4}\\
&&\mbox{correlation-risk}=E\left[\int_0^T G_{s-}dh^{(E)}_s(m^{\perp},\mathbb F)\right],\label{Case5}\\
&&\mbox{cost-of-leaving-earlier}=E\left[\int_0^T{{1-G_{s-}}\over{2}}{\widetilde\lambda}^{tr}_s c_s {\widetilde\lambda}_s dA_s\right].\label{Case6}
\end{eqnarray}
\end{theorem}
\subsection{Proofs of Theorems \ref{optimalportfoliogen}, \ref{Difference4u} and \ref{Proposition4DifferenceU} and Proposition \ref{Corollary5.3}}\label{Subsection4Proofs}
This subsection focuses on proving the  results of the previous subsection, and is divided into four sub-subsections.
\subsubsection{Proof of Theorem \ref{optimalportfoliogen}}
To prove this theorem, we start by deriving the predictable characteristics for $(S^{\tau},\mathbb G)$, that we denote by $(b^{\mathbb G}, c^{\mathbb G}, F^{\mathbb G}, A^{\mathbb G})$,  as follows.
\begin{eqnarray}\label{Charac4G}
&& b^{\mathbb G}:=b+c\beta+\int h(x)(f_m(x)-1)F(dx),\ \mu^{\mathbb G}:=I_{\Lbrack0,\tau\Lbrack}\star\mu,\ c^{\mathbb G}:=c\nonumber\\
&& d\nu^{\mathbb G}:=I_{\Lbrack0,\tau\Lbrack}f_{m}d\nu,\quad F^{\mathbb G}(dx):=I_{\Lbrack0,\tau\Lbrack}f_{m}(x)F(dx),\quad A^{\mathbb G}:=A^{\tau}.
\end{eqnarray}
Thus, by directly  applying Theorem \ref{LemmaCrucial}  to the model $(S^{\tau},\mathbb G)$, we deduce the equivalences between the existence of the log-optimal portfolio $\widetilde\theta^{\mathbb G}$ for the model, the existence of $\varphi\in {\cal L}(S^{\tau},\mathbb G)$ satisfying 
\begin{eqnarray}
&&(\theta-\varphi)^{tr}(b^{\mathbb G}-c^{\mathbb G}\varphi)+ \int \left( {{(\theta-\varphi)^{tr}x}\over{1+{\varphi}^{tr}x}}-(\theta-\varphi)^{tr}h(x)\right)F^{\mathbb G}(dx)\leq 0,\label{C3forStau}\hskip 1cm\\
&&E\left[ V^{\mathbb G}_T+{1\over{2}}(\varphi^{tr}c^{\mathbb G}\varphi\is A^{\mathbb G})_T+({\cal K}_{log}(\varphi^{tr}x)\star \nu^{\mathbb G})_T\right]<+\infty,\label{C1forStau}\\
&&  V^{\mathbb G}:=\Big\vert \varphi^{tr}b^{\mathbb G}-\varphi^{tr}c^{\mathbb G}\varphi+\int \left[{{\varphi^{tr}x}\over{1+\varphi^{tr}x}}-\varphi^{tr}h(x)\right] F^{\mathbb G}(dx)\Big\vert\is A^{\mathbb G}\label{VG}
\end{eqnarray}
for any bounded $\theta\in  {\cal L}(S^{\tau},\mathbb G)$, and the existence of num\'eraire portfolio rate $\widetilde\varphi^{\mathbb G}$ satisfying (\ref{C3forStau}). Furthermore, 
\begin{eqnarray}
&&{\widetilde Z}^{\mathbb G}={\cal E}({\widetilde\varphi}^{\mathbb G}\is S^{\tau})^{-1}\in  {\cal D}_{log}(S^{\tau}, \mathbb G),\quad \varphi=\varphi^{\mathbb G},\label{Ztilde1}\\
 &&{\widetilde Z}^{\mathbb G}={\cal E}({\widetilde K}^{\mathbb G}){\cal E}(-\widetilde V^{\tau}),\ {\widetilde K}^{\mathbb G}:=-\varphi\is S^{c,\mathbb G}+{{-{\widetilde\Gamma}^{\mathbb G}\varphi^{tr}x}\over{1+\varphi^{tr}x}}\star(\mu^{\tau}-\nu^{\mathbb G}),\label{Ztilde2}\\
 &&{\widetilde\Gamma}^{\mathbb G}_t:=\left(1+\int (f^{(op,\mathbb G)}_t(x)-1)\nu^{\mathbb G}(\{t\},dx)\right)^{-1}, f^{(op, \mathbb G)}_t(x):={1\over{1+\varphi^{tr}_t x}},\hskip 1cm\label{GammaG}
\end{eqnarray}
Thanks to Lemma \ref{PortfolioGtoF}, we deduce the existence of $\widetilde\varphi\in {\cal L}(S,\mathbb F)\cap L(S,\mathbb F)$ and 
\begin{eqnarray}\label{PhiEquality}
\varphi I_{\Lbrack0,\tau\Lbrack}=\widetilde\varphi I_{\Lbrack0,\tau\Lbrack},\quad P\otimes A\mbox{-a.e.}.\end{eqnarray} Therefore, by inserting this and  (\ref{Charac4G}) in (\ref{VG}) and (\ref{GammaG}), we conclude that 
\begin{eqnarray}\label{Equality4Gamma}
V^{\mathbb G}={\widetilde V}^{\tau},\quad {\widetilde\Gamma}^{\mathbb G}I_{\Lbrack0,\tau\Lbrack}={\widetilde\Gamma}I_{\Lbrack0,\tau\Lbrack},\quad   f^{(op,\mathbb G)}I_{\Lbrack0,\tau\Lbrack}=f^{(op)}I_{\Lbrack0,\tau\Lbrack},\end{eqnarray}
where $ \widetilde\Gamma$ and $f^{(op)}$are given by (\ref{GammaTilde}). Thus, by inserting (\ref{PhiEquality}), (\ref{Equality4Gamma}) and  (\ref{Charac4G}) in (\ref{C3forStau}) and (\ref{C1forStau}) and using Lemma \ref{PortfolioGtoF}-(d), we deduce that both (\ref{Cond4ptimalityG}) and (\ref{Cond4integrabilityG}) hold for $\widetilde\varphi$. Hence, (\ref{LogOPrate4G}), and the equivalence  between assertions (a), (b) and (c) follow immediately. By inserting  (\ref{PhiEquality}) in (\ref{Ztilde1}), we obtain the first equality in (\ref{OptimalDualLoggeneral}), while the rest of the proof focuses in proving the second equality of  (\ref{OptimalDualLoggeneral}). To this end, thanks to Theorem \ref{GeneralDeflators}-(b), we use (\ref{Ztilde2}) after inserting (\ref{PhiEquality}) and (\ref{Equality4Gamma}), and we look for ${\widetilde K}^{\mathbb F}\in {\cal M}_{0,loc}(\mathbb F)$ such that 
\begin{eqnarray}\label{fromKG2KF}
-\widetilde\varphi\is S^{c,\mathbb G}+{{-{\widetilde\Gamma}\widetilde\varphi^{tr}x}\over{1+\widetilde\varphi^{tr}x}}\star(\mu^{\tau}-\nu^{\mathbb G})={\cal T}\left({\widetilde K}^{\mathbb F}-{1\over{G_{-}}}\is m\right).
\end{eqnarray} 
By combining the fact that two local martingales are equal if and only if their continuous martingale parts are equal and their jumps are also equal,  the fact that ${\cal T}(S^c)=S^{(c,\mathbb G)}$ defined in (\ref{Charac4G}), and the fact that 
\begin{eqnarray*}
\Delta {\cal T}(X)={{G_{-}}\over{\widetilde G}}\Delta X I_{\Lbrack0,\tau\Lbrack}\quad\mbox{and}\ \Delta {\widetilde K}^{\mathbb G}=\left({{{\widetilde\Gamma}}\over{1+{\widetilde\varphi}^{tr}\Delta S}}-1\right)I_{\Lbrack0,\tau\Lbrack},\end{eqnarray*}
we deduce that (\ref{fromKG2KF}) is equivalent to 
\begin{eqnarray}
\hskip -0.55cm({\widetilde K}^{\mathbb F})^c=-\widetilde\varphi\is S^{c}+{1\over{G_{-}}}\is m^c,\ {{\widetilde\Gamma}\over{1+\widetilde\varphi^{tr}\Delta S}}I_{\Lbrack0,\tau\Lbrack}={{G_{-}}\over{\widetilde G}}\left(\Delta {\widetilde K}^{\mathbb F}+1\right)I_{\Lbrack0,\tau\Lbrack}.\label{fromKG2KFbis}
\end{eqnarray} 
By taking the $\mathbb F$-optional projection on both sides above, we get
\begin{eqnarray*}
\Delta{\widetilde K}^{\mathbb F}={{{\widetilde G}\widetilde\Gamma}\over{G_{-}(1+\widetilde\varphi^{tr}\Delta S)}}-1.
\end{eqnarray*}
Thanks to ${\widetilde G}=G_{-}(f_m(\Delta S)+g_m(\Delta S))$ on $(\Delta S\not=0)$ and $\Delta m={\widetilde G}-G_{-}$, the above equality is equivalent to
\begin{eqnarray}\label{DeltaKF}
\Delta{\widetilde K}^{\mathbb F}=-{{{\widetilde\Gamma}g_m(\Delta S)\widetilde\varphi^{tr}\Delta S}\over{1+\widetilde\varphi^{tr}\Delta S}}-{{{\widetilde\Gamma}f_m(\Delta S)\widetilde\varphi^{tr}\Delta S}\over{1+\widetilde\varphi^{tr}\Delta S}}+{{{\Delta m}{\widetilde\Gamma}}\over{G_{-}}}+{\widetilde\Gamma}-1.
\end{eqnarray}
Remark that it is not difficult to check that the two processes
\begin{eqnarray*}
 {{{\widetilde\Gamma}g_m\widetilde\varphi^{tr}x}\over{1+\widetilde\varphi^{tr}x}}\star\mu\quad \mbox{and}\quad-{{{\widetilde\Gamma}f_m\widetilde\varphi^{tr}x}\over{1+\widetilde\varphi^{tr}x}}\star(\mu-\nu)\end{eqnarray*}
 are well defined $\mathbb F$-local martingales (see Theorem \ref{tmgviacharacteristics} for the conditions that guarantree their existence), and 
 \begin{eqnarray*}
&& {{{\Delta m}{\widetilde\Gamma}}\over{G_{-}}}=\Delta\left({\widetilde\Gamma}G_{-}^{-1}\is m\right),\ -{{{\widetilde\Gamma}g_m(\Delta S)\widetilde\varphi^{tr}\Delta S}\over{1+\widetilde\varphi^{tr}\Delta S}}=\Delta\left(-{{{\widetilde\Gamma}g_m\widetilde\varphi^{tr}x}\over{1+\widetilde\varphi^{tr}x}}\star\mu\right),\\
&&-{{{\widetilde\Gamma}f_m(\Delta S)\widetilde\varphi^{tr}\Delta S}\over{1+\widetilde\varphi^{tr}\Delta S}}+{\widetilde\Gamma}-1=\Delta\left(-{{{\widetilde\Gamma}f_m\widetilde\varphi^{tr}x}\over{1+\widetilde\varphi^{tr}x}}\star(\mu-\nu)\right).\end{eqnarray*}
Thus, by combining theses facts with $\left({\widetilde\Gamma}G_{-}^{-1}\is m\right)^c=G_{-}^{-1}\is m^c$, the first equality in (\ref{fromKG2KFbis}) and (\ref{DeltaKF}), we deduce that ${\widetilde K}^{\mathbb F}$ is given by (\ref{K-F}), and the proof of the theorem is completed. \qed
 \subsubsection{Proof of Proposition \ref{Corollary5.3}.}
 Remark that $\tau$ is non-correlated to $S$ if and only if 
\begin{eqnarray}\label{Beta0F1}
c\beta=0\quad dP\otimes dA\mbox{-a.e.}\quad f_m(t,x)=1\quad P\otimes dA\otimes F\mbox{-a.e..}\end{eqnarray}
1) Here we prove assertion (a). Then thanks to Theorem \ref{LemmaCrucial} and to the assumption ${\cal D}(S,\mathbb F)\not=\emptyset$, the two num\'eraire portfolios with rates $\widetilde\varphi$ and $\widetilde\lambda$, for $(S^{\tau},\mathbb G)$ and $(S, \mathbb F)$ respectively, solve the following inequality-equation 
\begin{eqnarray}\label{OpmisationInequality}
(\theta-{\widetilde\theta})^{tr}(b-c{\widetilde\theta})+\int (\theta-\widetilde\theta)^{tr}({{x}\over{1+{\widetilde\theta}^{tr}x}}-h(x))F(dx)\leq 0,\ \theta\in {\cal L}(S,\mathbb F).
\hskip 0.75cm\end{eqnarray}
Thus, in virtue of the lemma below, whose proof is relegated to Appendix \ref{Section4Corollary5.3}, 
\begin{lemma}\label{Lemma4Uniqueness}There is at most one $\widetilde\theta\in {\cal L}(S,\mathbb F)$ satisfying (\ref{OpmisationInequality}). Herein, $\psi$ and $\varphi$ --elements of ${\cal L}(S,\mathbb F)$-- are said to be equal if $\psi(1+\vert\psi\vert)^{-1}\is S=\varphi(1+\vert\varphi\vert)^{-1}\is S$.
\end{lemma}
we deduce that both portfolios coincides on $\Lbrack0,\tau\Lbrack$. This ends the proof of assertion (a) of the proposition.\\
2) In virtue of assertion (a) of the proposition, (\ref{Beta0F1}), and Theorem \ref{optimalportfoliogen}-(c), we claim that the log-optimal portfolio for $(S^{\tau},\mathbb G)$ exists if and only if the num\'eraire portfolio rate $\widetilde\lambda$ satisfies 
\begin{eqnarray*}
E\left[G_{-}\is {\widetilde V}^{\mathbb F}_T+G_{-}{\widetilde\lambda}^{tr}c{\widetilde\lambda}\is A_T+G_{-}{\cal K}_{log}({\widetilde\lambda}^{tr}x)\star\nu_T\right]<+\infty,
\end{eqnarray*}
where $ {\widetilde V}^{\mathbb F}$ is given by 
\begin{eqnarray*}
 {\widetilde V}^{\mathbb F}:=\left({\widetilde\lambda}^{tr}b-{\widetilde\lambda}^{tr}c\widetilde\lambda\right)\is A +\left({{{\widetilde\lambda}^{tr}x}\over{1 + \widetilde\lambda^{tr} x }}  -{\widetilde\lambda}^{tr}h\right)\star\nu.\end{eqnarray*}
 Therefore, due to $G_{-}\leq 1$, the above condition is implied by the existence of log-optimal portfolio for $(S, \mathbb F)$, see Theorem \ref{LemmaCrucial} applied to the model $(S,\mathbb F)$. This ends the proof of assertion (b).\\
3) Suppose that $S$ is continuous. Thus, both inequalities-equations for $\widetilde\varphi$ and $\widetilde\lambda$, given by (\ref{Cond4ptimalityG}) and (\ref{C6forX}) applied to  $(S,\mathbb F)$ respectively, become
\begin{eqnarray*}
(\theta-\widetilde\varphi)^{tr}(b+c(\beta-\widetilde\varphi))\leq 0,\quad (\theta-\widetilde\lambda)^{tr}(b-c\widetilde\lambda)\leq 0,\quad\forall\ \theta\in {\cal L}(S,\mathbb F).\end{eqnarray*}By combining this with the fact that any $\mathbb F$-predictable process belongs to ${\cal L}(S,\mathbb F)$, we conclude that $\widetilde\varphi=\beta+\widetilde\lambda$, and the proof of assertion (c.1) is completed. To prove assertion (c.2), it is enough to remark that assertion (c.1) implies that $\widetilde V\equiv 0$ and hence (\ref{Cond4integrabilityG}) becomes $E\left[\int_0^T G_{s-}{\widetilde\varphi}^{tr}_sc_s{\widetilde\varphi}_s dA_s\right]<+\infty$. Assertion (c.3) follows immediately from combining assertions (c.1) and (c.2). This ends the proof of the proposition.\qed
 \subsubsection{Proof of Theorem \ref{Difference4u}}
Direct calculations, see also \cite{ChoulliStricker2007} for similar details, for any $\mathbb H$-local martingale $X$ such that  such that $\Delta X>-1$, we have 
\begin{eqnarray}
-\ln({\cal E}(X))=-X +H^{(0)}(X, \mathbb H).\label{Hellinger4X}\end{eqnarray}
{\bf Part 1.} Here we prove the equality (\ref{Delta(S,tau)2}). To this end, by applying Theorem \ref{LemmaCrucial} to the model $(S,\mathbb F)$ and using Proposition \ref{Hzero4Log(Z)}, we derive
\begin{eqnarray}
&&u_T(S, \mathbb F)=E\left[\ln({\cal E}_T(\widetilde\lambda\is S))\right]=E\left[-\ln({\cal E}_T(-{\widetilde V}^{\mathbb F})\right]+E\left[-\ln({\cal E}_T({\widetilde L}^{\mathbb F}))\right]\nonumber\\
&&=E\left[{\widetilde V}^{\mathbb F}_T+\sum_{0<s\leq T}(-\Delta{\widetilde V}^{\mathbb F}_s-\ln(1-\Delta{\widetilde V}^{\mathbb F}_s))+H^{(0)}_T({\widetilde L}^{\mathbb F}, \mathbb F)\right]\nonumber\\
&&=E\left[{\widetilde{\cal H}}_T({\mathbb F})\right]=E\left[({\widetilde G}\is {\widetilde{\cal H}}({\mathbb F}))_T\right]+E\left[((1-\widetilde G)\is{\widetilde{\cal H}}({\mathbb F}))_T\right].\label{Htilde4F1}\\
&&=E\left[{\widetilde\lambda}^{tr}(b-{1\over{2}}c\widetilde\lambda)\is A_T+\left(\ln(1+{\widetilde\lambda}^{tr}x)-{\widetilde\lambda}^{tr}h\right)\star\nu_T\right].\label{Htilde4F2}
\end{eqnarray}
Thanks to the duality (\ref{OptimalDualLoggeneral}) of Theorem \ref{optimalportfoliogen} and Proposition \ref{Hzero4Log(Z)}, we derive 
\begin{eqnarray}
&&u_T(S^{\tau},\mathbb G)=E\left[\ln({\cal E}_T(\widetilde\varphi\is S^{\tau}))\right]=E\left[-\ln({\widetilde Z}^{\mathbb G}_T)\right]\nonumber\\
&&=E\left[-\ln({\cal E}_{\tau\wedge T}(-\widetilde V))\right]+E\left[-\ln({\cal E}_T({\widetilde K}^{\mathbb G}))\right]\nonumber\\
&&=E\left[{\widetilde V}_{\tau\wedge T}+\sum_{0<s\leq T\wedge\tau}(-\Delta{\widetilde V}_s-\ln(1-\Delta{\widetilde V}_s))+H^{(0)}_T({\widetilde K}^{\mathbb G}, \mathbb G)\right]\label{Equa1theorem5.1}\hskip 0.5cm
\end{eqnarray}
Remark that, in virtue of (\ref{K-G}) and (\ref{GammaTilde}),  we get
\begin{eqnarray*}
\Delta{\widetilde K}^{\mathbb G}=\left({{{\widetilde\Gamma}}\over{1+{\widetilde\varphi}^{tr}\Delta S}}-1\right)I_{\Lbrack0,\tau\Lbrack},
\end{eqnarray*}
and hence by combining this with Definition \ref{Hellinger}, we derive 
\begin{eqnarray*}
&&H^{(0)}({\widetilde K}^{\mathbb G}, \mathbb G)\\
&&={1\over{2}}{\widetilde\varphi}^{tr}c{\widetilde\varphi}\is A^{\tau}+\sum\left({{{\widetilde\Gamma}}\over{1+{\widetilde\varphi}^{tr}\Delta S}}-1-\ln\left({{{\widetilde\Gamma}}\over{1+{\widetilde\varphi}^{tr}\Delta S}}\right)\right)I_{\Lbrack0,\tau\Lbrack}\\
&&={1\over{2}}{\widetilde\varphi}^{tr}c{\widetilde\varphi}\is A^{\tau}+\sum\left({\widetilde\Gamma}-1-\ln({\widetilde\Gamma})\right)I_{\Lbrack0,\tau\Lbrack}+\left({{-{\widetilde\Gamma}{\widetilde\varphi}^{tr}x}\over{1+{\widetilde\varphi}^{tr}x}}+\ln(1+{\widetilde\varphi}^{tr}x)\right)\star\mu^{\tau}.
\end{eqnarray*}
By inserting this in (\ref{Equa1theorem5.1}) and using afterwards (\ref{GammaTilde}), $1-\Delta{\widetilde V}=1/{\widetilde\Gamma}$ and the fact that $(1-{\widetilde\Gamma}){\widetilde\varphi}^{tr}x(1+ {\widetilde\varphi}^{tr}x)^{-1}f_m\star\nu=-\sum(1-\widetilde\Gamma)^2/{\widetilde\Gamma}$, we obtain
\begin{eqnarray}
&&u_T(S^{\tau}, \mathbb G)\nonumber\\
&&=E\left[G_{-}\is\left({\widetilde V}+\sum_{0<s\leq\cdot}(-\Delta{\widetilde V}_s-\ln(1-\Delta{\widetilde V}_s))+\sum_{0<s\leq\cdot}({\widetilde\Gamma}_s-1-\ln({\widetilde\Gamma}_s))\right)_T\right]\nonumber\\
&&+E\left[{{G_{-}}\over{2}}{\widetilde\varphi}^{tr}c{\widetilde\varphi}\is A_T\right]+E\left[G_{-}\left({{-{\widetilde\Gamma}{\widetilde\varphi}^{tr}x}\over{1+{\widetilde\varphi}^{tr}x}}+\ln(1+{\widetilde\varphi}^{tr}x)\right)f_m\star\nu_T\right]\nonumber\\
&&=E\left[(G_{-}\is{\widetilde V})_T+{{G_{-}}\over{2}}{\widetilde\varphi}^{tr}c{\widetilde\varphi}\is A_T+G_{-}\left({{-{\widetilde\varphi}^{tr}x}\over{1+{\widetilde\varphi}^{tr}x}}+\ln(1+{\widetilde\varphi}^{tr}x)\right)f_m\star\nu_T\right]\nonumber\\
&&=E\left[G_{-}{\widetilde\varphi}^{tr}(b+c(\beta- {{\widetilde\varphi}\over{2}}))\is A_T+G_{-}\left(f_m\ln(1+{\widetilde\varphi}^{tr}x)-{\widetilde\varphi}^{tr}h\right)\star\nu_T\right]
\label{Equa2theorem5.1}\end{eqnarray}
The last equality follows from using (\ref{VtildeG}). Denote by ${\widetilde\Phi}$ the functional defined on ${\cal L}(S,\mathbb F)\times{\cal L}(S,\mathbb F)$ as follows.
\begin{eqnarray}\label{PhiTilde}
{\widetilde\Phi}(\lambda,\varphi)&&:=(\varphi-\lambda)^{tr}b+(\varphi-\lambda)^{tr}c\beta-{1\over{2}}\varphi^{tr}c\varphi+{1\over{2}}\lambda^{tr}c\lambda\\
&&+\int\left(f_m(x)\ln\left({{1+\varphi^{tr}x}\over{1+\lambda^{tr}x}}\right)-(\varphi-\lambda)^{tr}h(x)\right)F(dx).\nonumber\end{eqnarray}
Then by using the convexity of both functions of $\varphi$, $\varphi^{tr}c\varphi$ and $-\ln(1+{\varphi}^{tr}x)$ and (\ref{Cond4ptimalityG}), we deduce the nonnegativity of the process ${\widetilde {\cal R}}$, i.e.
\begin{eqnarray}\label{Rpositive}
&&{\widetilde {\cal R}}={\widetilde\Phi}(\widetilde\varphi,\widetilde\lambda)\geq 0,\end{eqnarray}
and in virtue of (\ref{Equa2theorem5.1}), (\ref{Htilde4F1}) and (\ref{Htilde4F2}), we get 
\begin{eqnarray*}
\Delta_T(S, \tau, \mathbb F)&&= u_T(S^{\tau}, \mathbb G)-u_T(S, \mathbb F)\\
&&=-E\left[(1-\widetilde G)\is {\widetilde{\cal H}}(\mathbb F))_T\right]+E\left[G_{-}{\widetilde\Phi}(\widetilde\varphi,\widetilde\lambda)\is A_T\right]\\
&&+E\left[G_{-}{{(f_m-1)({\widetilde\lambda}^{tr}x)}\over{1+{\widetilde\lambda}^{tr}x}}\star\nu_T+G_{-}{\widetilde\lambda}^{tr}c\beta\is A_T\right].\nonumber\end{eqnarray*}
Thus, by combining this latter equality with  
\begin{eqnarray*}
E\left[\langle {\widetilde L}^{\mathbb F}, m\rangle^{\mathbb F}_T\right]=-E\left[G_{-}{{(f_m-1)({\widetilde\lambda}^{tr}x)}\over{1+{\widetilde\lambda}^{tr}x}}\star\nu_T+G_{-}{\widetilde\lambda}^{tr}c\beta\is A_T\right]
\end{eqnarray*}
which follows from direct calculations, we deduce that  (\ref{Delta(S,tau)2}) holds.\\
{\bf Part 2.} Here we prove the equality (\ref{Delta(S,tau)1}). To this end, we apply (\ref{Hellinger4X}) to $({\widetilde K}^{\mathbb F}, \mathbb F)$ and $(G_{-}^{-1}\is m, \mathbb F)$,  and we use the notation (\ref{Htilde(G)}) afterwards to get 
\begin{eqnarray}
&&-\ln\left({\cal E}(-{\widetilde V}^{\tau})\right)-\ln\left({\cal E}({\widetilde K}^{\mathbb F})^{\tau}\right)+\ln\left({\cal E}(G_{-}^{-1}\is m)^{\tau}\right)\nonumber\\
&&={\widetilde V}^{\tau}+\sum\left(-\Delta{\widetilde V}^{\tau}-\ln(1-\Delta{\widetilde V}^{\tau})\right)-({\widetilde K}^{\mathbb F})^{\tau} +H^{(0)}({\widetilde K}^{\mathbb F}, \mathbb F)^{\tau}\nonumber\\
&&+{1\over{G_{-}}}\is m^{\tau}-H^{(0)}\left({1\over{G_{-}}}\is m, \mathbb F\right)^{\tau}\nonumber\\
&&={\widetilde{\cal H}}(\mathbb G)^{\tau}-({\widetilde K}^{\mathbb F})^{\tau}+{1\over{G_{-}}}\is m^{\tau}-H^{(0)}\left({1\over{G_{-}}}\is m, \mathbb F\right)^{\tau}.\label{Log(KG)}
\end{eqnarray}
Thus, by taking the expectation on both sides and using (\ref{OptimalDualLoggeneral}), we derive 
\begin{eqnarray}
&&u_T(S^{\tau}, \mathbb G)=E\left[\ln({\cal E}_T(\widetilde\varphi\is S^{\tau}))\right]=E\left[-\ln({\cal E}_{T\wedge\tau}(-\widetilde V)\right]+E\left[-\ln({\cal E}_T({\widetilde K}^{\mathbb G}))\right]\nonumber\\
&&=E\left[({\widetilde G}\is {\widetilde{\cal H}}({\mathbb G}))_T-\langle{\widetilde K}^{\mathbb F}, m\rangle^{\mathbb F}_T+{1\over{G_{-}}}\is m_{\tau\wedge T}-H^{(0)}\left({1\over{G_{-}}}\is m, \mathbb F\right)_{T\wedge\tau}\right].\label{Equation5.400}
\end{eqnarray}
Therefore, by using (\ref{HellingerE}) and inserting (\ref{equa401}) in (\ref{Equation5.400}), we obtain
\begin{eqnarray*}\label{Equa546}
u_T(S^{\tau},  \mathbb G)= E\left[({\widetilde G}\is {\widetilde{\cal H}}({\mathbb G}))_T-\langle{\widetilde K}^{\mathbb F}, m\rangle^{\mathbb F}_T\right]+ {\cal H}_{\mathbb G}\left(P\big|{\widetilde Q}_T\right)
\end{eqnarray*}
Thus, by combining this with (\ref{Htilde4F1}), (\ref{Delta(S,tau)1}) follows immediately. Hence, in virtue of (\ref{Rpositive}), the proof of the theorem will be complete as soon as we prove
\begin{eqnarray}\label{Htilde(F,G)}
{\cal W}(\mathbb F, \mathbb G):=\left({\widetilde G}\is\left({\widetilde{\cal H}}(\mathbb F)-{\widetilde{\cal H}}(\mathbb G)\right)\right)^{p,\mathbb F}+\langle{\widetilde K}^{\mathbb F}-{\widetilde L}^{\mathbb F}, m\rangle^{\mathbb F}\in {\cal A}^+(\mathbb F).
\end{eqnarray}
On the one hand, similar calculations and arguments as in (\ref{Log(KG)}) applied to 
\begin{eqnarray}\label{ZbarG}
{\overline Z}^{\mathbb G}:={\cal E}({\widetilde L}^{\mathbb F})^{\tau}{\cal E}(-{\widetilde V}^{\mathbb F})^{\tau}/{\cal E}(G_{-}^{-1}\is m)^{\tau},\end{eqnarray}
lead to 
\begin{eqnarray*}
-\ln({\overline Z}^{\mathbb G})={\widetilde{\cal H}}(\mathbb F)^{\tau}-({\widetilde L}^{\mathbb F})^{\tau}+{1\over{G_{-}}}\is m^{\tau}-H^{(0)}\left({1\over{G_{-}}}\is m, \mathbb F\right)^{\tau}.\end{eqnarray*}
Then by combining this equality with (\ref{Log(KG)}), we obtain 
\begin{eqnarray*}
-\ln\left({\widetilde Z}^{\mathbb G}/{\overline Z}^{\mathbb G}\right)={\widetilde{\cal H}}(\mathbb G)^{\tau}-{\widetilde{\cal H}}(\mathbb F)^{\tau}-({\widetilde K}^{\mathbb F})^{\tau}+({\widetilde L}^{\mathbb F})^{\tau}.
\end{eqnarray*}
On the other hand, by using Jenson's inequality and the facts that $1/{\widetilde Z}^{\mathbb G}={\cal E}(\widetilde\varphi\is S^{\tau})$, ${\overline Z}^{\mathbb G}\in {\cal D}(S^{\tau},\mathbb G)$ and both processes ${\widetilde{\cal H}}(\mathbb G)$ and ${\widetilde{\cal H}}(\mathbb F)$ are $\mathbb F$-optional, we deduce that $-\ln\left({\widetilde Z}^{\mathbb G}/{\overline Z}^{\mathbb G}\right)$ is a $\mathbb G$-submartingale, and hence ${\cal W}(\mathbb F, \mathbb G)$ is nondecreasing and $\mathbb F$-predictable
This ends the proof of the theorem.\qed
\subsubsection{Proof of Theorem  \ref{Proposition4DifferenceU}} This proof has four parts, where the four assertions are proved respectively.\\
1) Here we prove assertion (a). It is clear that the ``correlation-risk" is null if and only if the process ${\cal W}(\mathbb F, \mathbb G)$ defined in (\ref{Htilde(F,G)}) is null, or equivalently the two deflators $\widetilde Z^{\mathbb G}$ and ${\overline Z}^{\mathbb G}$ defined in (\ref{ZbarG}) are equal. This is obviously equivalent, due to the uniqueness of the Doob-Meyer decomposition in $\mathbb G$, to  
\begin{eqnarray*}\label{DoobG}
({\widetilde V}^{\mathbb F})^{\tau}={\widetilde V}^{\tau}\quad\mbox{and}\quad {\cal E}({\widetilde K}^{\mathbb F})^{\tau}= {\cal E}({\widetilde L}^{\mathbb F})^{\tau}.
\end{eqnarray*}
Hence, thanks to the assumption $G>0$, the above equalities are equivalent to 
\begin{eqnarray}\label{DoobG}
{\widetilde V}^{\mathbb F}={\widetilde V}\quad\mbox{and}\quad{\widetilde K}^{\mathbb F}={\widetilde L}^{\mathbb F}.
\end{eqnarray}
Then, in virtue of the uniqueness of Jacod's decomposition of Theorem \ref{tmgviacharacteristics}, we conclude that  the ``correlation-risk" is null if and only if 
\begin{eqnarray}\label{JacodG}
c(\beta-\widetilde\varphi)=-c{\widetilde\lambda},\ m^{\perp}\equiv 0,\ g_m=0\ \mbox{and}\ {{f_m}\over{1+{\widetilde\varphi}^{tr}x}}={1\over{1+{\widetilde\lambda}^{tr}x}}\ M^P_{\mu}-a.e.\hskip 0.75cm
\end{eqnarray} 
Therefore, in virtue of this latter equivalence, we deduce that   ``correlation-risk" is null and  $\widetilde\varphi=\widetilde\lambda$ is equivalent to 
\begin{eqnarray*}
c\beta\equiv 0,\  m^{\perp}\equiv 0,\ g_m=0\ \mbox{and}\ f_m=1\ M^P_{\mu}-a.e.\end{eqnarray*}
This is equivalent to $m\equiv m_0$, i.e. $\tau$ is a pseudo-stopping time. This proves the first statement in assertion (a). Furthermore, due to $m\equiv m_0$, we get $\langle{\widetilde L}^{\mathbb F}, m\rangle^{\mathbb F}\equiv 0$ and hence the NP$(\mathbb F)$-correlation factor is null. Thus, by inserting all these in (\ref{Delta(S,tau)1}), we get (\ref{Case1}), and the proof of assertion (a) is completed.\\
2) This part addresses assertion (b). Suppose that $\tau$ is non-correlated to $S$,  which is equivalent to $c\beta=0$ $P\otimes A$-a.e. and $f_m=1$ $P\otimes A\otimes F$-a.e.. Then thanks to Lemma \ref{Lemma4Uniqueness}, we deduce that $\widetilde\lambda=\widetilde\varphi$, and hence we derive 
\begin{eqnarray*}
\ \langle{\widetilde L}^{\mathbb F}, m\rangle^{\mathbb F}\equiv 0\quad \mbox{and}\quad {\widetilde {\cal R}}\equiv 0.\end{eqnarray*}
This implies that both factors of ``NP($\mathbb F)$-correlation" and  ``num\'eraire-change-premium" factor are null, and hence ``information-premium" coincides with "correlation-risk". By inserting all these in (\ref{Delta(S,tau)1}), we obtain again (\ref{Case1}) and the proof of assertion (b) is completed.\\
3) Here we prove assertions (c) and (d). On the one hand, when $\widetilde\lambda$ coincides with $\widetilde\varphi$, then $\widetilde{\cal R}\equiv 0$ follows from (\ref{NumeraireChanegPremuim}). On the other hand, using Taylor's expansion and (\ref{Cond4ptimalityG}) for $\theta=\widetilde\lambda$, we derive 
\begin{eqnarray*}
 {\widetilde{\cal R}}\geq (\widetilde\lambda-\widetilde\varphi)^{tr}c(\widetilde\lambda-\widetilde\varphi)+\int{{((\widetilde\lambda-\widetilde\varphi)^{tr}x)^2}\over{\max((1+\widetilde\lambda^{tr}x)^2, (1+\widetilde\varphi^{tr}x)^2)}}F(dx).
\end{eqnarray*}
Therefore, we deduce that $ {\widetilde{\cal R}}$ is null iff $c\widetilde\lambda=c\widetilde\varphi$ $P\otimes A$-a.e. and $\widetilde\lambda^{tr}x=\widetilde\varphi^{tr}x$ $P\otimes A\otimes F$-a.e.. Therefore, assertion (c) follows from combining this latter claim and the fact that, due to $G>0$ and in virtue of (\ref{Delta(S,tau)2}),   
\begin{eqnarray}\label{Rzero}
\mbox{ the ``num\'eraire-change-premium" is null iff}\quad  {\widetilde{\cal R}}\equiv0\quad P\otimes A\mbox{-.a.e..}\end{eqnarray}
 The rest fo this part proves assertion (d). Suppose $S\in {\cal M}_{loc}(\mathbb F)$. Then we get $\widetilde\lambda\equiv 0$ and ${\widetilde L}^{\mathbb F}\equiv 0$. As a result, we deduce that $u_T(S, \mathbb F)=0$ and $\Delta_T(S, \tau, \mathbb F)=u_T(S^{\tau}, \mathbb G)\geq 0$. This ends the proof of assertion (d). \\
4) Suppose $S$ is continuous. Then thanks to Proposition \ref{Corollary5.3}-(c), we deduce the following equalities
\begin{eqnarray}\label{equalities00}
{\widetilde K}^{\mathbb F}=(\beta-{\widetilde\varphi})\is S^c+m^{\perp}={\widetilde L}^{\mathbb F}+m^{\perp},\ {\widetilde V}^{\mathbb F}={\widetilde V}=0.\end{eqnarray}
Therefore, direct calculations on Hellinger processes, see also \cite{ChoulliStricker2005,ChoulliStricker2006,ChoulliStricker2007} for more details about this fact, we derive 
\begin{eqnarray*}
&&H^{(0)}({\widetilde K}^{\mathbb F}, P)=H^{(0)}({\widetilde L}^{\mathbb F}, P)+H^{(0)}(m^{\perp}, P),\ {\widetilde G}=G_{-}(1+\Delta m^{\perp}),\\
&&-{\widetilde G}\is H^{(0)}(m^{\perp}, P)=G_{-}\is H^{(E)}(m^{\perp}, P)-G_{-}\is [m^{\perp}, m^{\perp}]\\
&&h^{(E)}(G_{-}^{-1}\is m, P)={1\over{2}}\beta^{tr}c\beta\is A+h^{(E)}(m^{\perp}, P)
\end{eqnarray*}
Thus, by combining these equalities with (\ref{Delta(S,tau)1}) and (\ref {equalities00}), assertion (e) follows immediately and the proof of the theorem is completed.  \qed
\subsection{The case when $(S,\mathbb F)$ is a jump-diffusion model}\label{section4JumpDifusionCase}

This subsection illustrates the main results of Section \ref{section3} and Subsection \ref{section4} on the case where the initial model $(S,\mathbb F)$ is a one-dimensional jump-diffusion model. Precisely, we suppose that a standard Brownian motion $W$ and a Poisson process $N$ with intensity $\lambda>0$ are defined on  the probability space $(\Omega, {\cal F}, P)$, the filtration $\mathbb F$ is the completed and right continuous filtration generated by $W$ and $N$. Consider a fixed horizon $T\in (0,+\infty)$, and suppose $S$ satisfies
\begin{equation}\label{SPoisson2}
S_t:=S_0 {\cal E} (X)_t,\   X_t: =(\sigma\is W)_t+(\zeta\is {N}^{\mathbb F})_t + \int_{0}^{t} \mu_s ds,\  {N_t}^{\mathbb F}:=N_t-\lambda t,
\end{equation}
and there exists a constant $\delta\in(0,+\infty)$ such that $\mu$, $\sigma$ and $\zeta$ are bounded $\mathbb F$-predictable processes satisfying
\begin{eqnarray}\label{parameters2}
\zeta>-1,\quad\min( \sigma, \vert\zeta\vert)\geq \delta,\ P\otimes dt\mbox{-a.e.}.
 \end{eqnarray}
Since $m$ is an $\mathbb F$-martingale, then there exists two $\mathbb F$-predictable processes $\varphi^{(m)}$ and $\psi^{(m)}$ such that  $\int_0^T \left((\varphi^{(m)}_s)^2 +\vert \psi^{(m)}_s\vert\right) ds<+\infty\ P\mbox{-a.s.}$, and
\begin{eqnarray}\label{model4tau2}
 G_{-}^{-1}\is m=\varphi^{(m)}\is W+(\psi^{(m)}-1)\is N^{\mathbb F}.\end{eqnarray}

  \begin{theorem}\label{OptimalDeflatorLogJD} Suppose $G>0$, $S$ be given by (\ref{SPoisson2})-(\ref{parameters2}), and consider
 \begin{equation}\label{thetaTilde}
 \widetilde{\theta}:=  \displaystyle\frac{\xi + sign(\zeta) \sqrt{\xi^2 + 4 \lambda \psi^{(m)}} }{2\sigma} - \frac{1}{\zeta},\ \mbox{where}\ \xi:= \frac{\mu-\lambda\zeta}{\sigma} + \varphi^{(m)}+ \frac{\sigma}{\zeta}, \end{equation}
 Then  $ \widetilde{\theta}\in{\cal L}(S, \mathbb F)\cap L(S, \mathbb F)$ is the num\'eraire portfiolio rate for $(S^{\tau}, \mathbb G)$, and the following assertions are equivalent.\\
{\rm{(a)}} The random time $\tau$, parametrized in $\mathbb F$ by $(\varphi^{(m)}, \psi^{(m)},G_{-})$, satisfies 
\begin{eqnarray}\label{ConditionHellinger}
E\left[\int_0^T G_{s-}\left[(\varphi^{(m)}_s)^2+\lambda\psi^{(m)}_s\ln(\psi^{(m)}_s)-\lambda\psi^{(m)}_s+\lambda\right]dt\right]<+\infty.\end{eqnarray}
{\rm{(b)}}   The solution to (\ref{dualproblem}) 
exists and is given by
\begin{equation}\label{ZtildeG}
 {\widetilde Z}^{\mathbb G}:={\cal E}({\widetilde K}^{\mathbb G}),\  {\widetilde K}^{\mathbb G}:= - \sigma\widetilde\theta\is {\cal T}({W}) -\frac{\psi^{(m)}\zeta{\widetilde\theta}}{1+{\widetilde\theta}\zeta} \is{\cal T}({N^{\mathbb F}}).\end{equation}
{\rm{(c)}}  $\widetilde{\theta}$ is the log-optimal portfolio rate for the model $(S^{T\wedge\tau},\mathbb G)$. \\
\end{theorem}
 \begin{proof} 
 For the model (\ref{SPoisson2})-(\ref{parameters2}), the predictable characteristics of Section 3 can be derived as follows. Let  $\delta_{a}(dx) $ be the Dirac mass at the point $a$. Then in this case we have $d=1$ and 
 \begin{eqnarray*}
&&\mu (dt,dx) = \delta_{\zeta_tS_{t-}}(dx)dN_t, \  \nu  (dt , dx)= \delta_{\zeta_tS_{t-}}(dx) \lambda dt,\ F_t(dx) = \lambda \delta_{\zeta_tS_{t-}}(dx), \\
&&A_t = t,\ c= (S_{-}\sigma)^2,\ b= (\mu - \lambda \zeta I_{\{ |\zeta|S_{-} > 1 \}}) S_{-}, \ (\beta, g_m,m^\perp)=({{\varphi^{(m)}}\over{S_{-}\sigma}},0,0).\end{eqnarray*}
As a result, the set 
\begin{eqnarray*}
{\cal L}_{(\omega, t)}(S,\mathbb F)&&:=\{\varphi\in\mathbb R\ \big|\ \varphi x>-1\ F_{(\omega, t)}(dx)-a.e.\}=\{\varphi\in\mathbb R\ \big|\ \varphi S_{-}\zeta>-1\}\\
&&=\left(-{1/(S_{-}\zeta)^+},{1/(S_{-}\zeta)^-}\right)\end{eqnarray*}
 is an open set in $\mathbb R$ (with the convention $1/0^+=+\infty$). Then the condition (\ref{Cond4ptimalityG}), characterizing the optimal portfolio $\widetilde\varphi$, becomes an equation as follows.
 \begin{eqnarray}
0 && = \mu-\lambda\zeta I_{\{\vert\zeta\vert>{1/S_{-}}\}}+S_{-}\sigma^2({{\varphi^{(m)}}\over{S_{-}\sigma}}-\theta)+\lambda{{\psi^{(m)}\zeta}\over{1+S_{-}\theta\zeta}}-\lambda\zeta  I_{\{\vert\zeta\vert\leq {1/S_{-}}\}} \nonumber\\
&& = \mu -\lambda \zeta+\sigma\varphi^{(m)} - S_{-}\sigma^2\theta + \frac{\psi^{(m)} \lambda\zeta}{1+\theta S_{-}\zeta} .\label{mainequation4levy}\end{eqnarray}
By putting $\varphi:=1+\theta S_{-}\zeta>0$, the above equation is equivalent to 
$$0=- {{\sigma^2}\over{\zeta}}\varphi^2 +[\mu -\lambda \zeta+\sigma\varphi^{(m)}+{{\sigma^2}\over{\zeta}} ]\varphi+ \psi^{(m)} \lambda\zeta,$$
which has always (since $\psi^{(m)}>0$) a unique positive solution given by 
$$\widetilde\varphi:={{\Gamma\zeta+\vert\zeta\vert\sqrt{\Gamma^2+4\sigma^2\lambda\psi^{(m)}}}\over{2\sigma^2}},\quad \Gamma:=\mu -\lambda \zeta+\sigma\varphi^{(m)}+{{\sigma^2}\over{\zeta}}.$$
Hence, we deduce that $\widetilde\lambda:={\widetilde{\theta}}/S_{-}$, where $\widetilde\theta$ is given by (\ref{thetaTilde}),  coincides with $(\widetilde{\varphi}-1)/(S_{-}\zeta)$, satisfies $ 1 + \zeta\widetilde{\theta}>0$, and hence it is the unique solution to (\ref{mainequation4levy}). It is also clear that $\widetilde{\theta}$ is $S$-integrable (or equivalently $\widetilde\lambda$ is $S$-integrable) due to the assumptions in (\ref{parameters2})-(\ref{model4tau2}). As a result,  the optimal wealth process is ${\cal E}(\widetilde\lambda\is S^{\tau})={\cal E}(\widetilde\theta\is X^{\tau})$ and assertions (a) and (b) follow immediately using the above analysis and Theorems \ref{optimalportfoliogen}.
\qed\end{proof}

\vspace*{0,5cm}
\centerline{\textbf{APPENDIX}}
\appendix
 \normalsize
\section{Some $\mathbb G$-properties versus those in $\mathbb F$}
Some results in the following lemma sounds new to us. 
\begin{lemma}\label{PortfolioGtoF}  Let $A$ is a nondecreasing and $\mathbb F$-predictable, and suppose that $G>0$. Then the following assertions hold.\\
{\rm{(a)}} For any $\mathbb G$-predictable process $\varphi^{\mathbb G}$, there exists an $\mathbb F$-predictable process $\varphi^{\mathbb F}$ such that
  $\varphi^{\mathbb G}=\varphi^{\mathbb F}$ on ${\Lbrack0,\tau\Lbrack}$. Furthermore, if $\varphi^{\mathbb G}>0$ (respectively $\varphi^{\mathbb G}\leq 1$), then $\varphi^{\mathbb F}>0$ (respectively $\varphi^{\mathbb F}\leq 1$). \\
{\rm{(b)}} For any $\theta\in{\cal  L}(S^{\tau},\mathbb G)$, there exists ${\varphi}\in {\cal L}(S,\mathbb F)$ such that $ {\varphi}={\theta} $ on $\Lbrack0,\tau\Lbrack$.\\
{\rm{(c)}} For any $\theta\in L(S^{\tau},\mathbb G)$, there exists ${\varphi}\in L(S,\mathbb F)$ such that $ {\varphi}={\theta} $ on $\Lbrack0,\tau\Lbrack$.\\
{\rm{(d)}} Let  $v$ be an $\mathbb F$-predictable process. Then $v I_{\Lbrack0,\tau\Lbrack}\leq 0$ $P\otimes A$-a.e. if and only if $v\leq 0$ $P\otimes A$-a.e..\\
{\rm{(e)}} Let $\varphi$ be a nonnegative and $\mathbb F$-predictable process. Then $\varphi<+\infty$ $P\otimes A$-a.e. on $\Lbrack0,\tau\Lbrack$ if and only if $\varphi<+\infty$ $P\otimes A$-a.e.\\ 
{\rm{(f)}}  Let $V$ be an $\mathbb F$-predictable and nondecreasing process that takes values in $[0,+\infty]$. If $V^{\tau}$ is $\mathbb G$-locally integrable, then $V$ is $\mathbb F$-locally integrable.
\end{lemma}
\begin{proof} Assertion (a) is a particular case of \cite[Lemma B.1]{ACDJ1} and assertion (b) can be found in \cite[Lemma A.1]{ChoulliYansori1}, while assertions (e) and (f) follow immediately from \cite[Proposition B.2-(c)-(f)]{ACDJ1}. Thus the rest of this proof focuses on proving assertions (c) and (d).\\
{\rm{(c)}}  Let  $\theta\in{  L}(S^{\tau},\mathbb G)$. Then  on the one hand, due to \cite[Theorem 1.16, or Remark 2.2-(h)]{Stricker}, this equivalent to the set 
\begin{eqnarray*}\label{Lzero1}
{\cal X}^{\mathbb G}:=\left\{ \sup_{t\geq 0}\vert (H\theta I_{\{\vert\theta\vert\leq n}\is S^{\tau})_t\vert\ \Big|\ H\ {\mathbb G}-\mbox{predictable}\ \vert H\vert\leq 1, n\geq 1\right\}\end{eqnarray*}
being bounded in  probability.  On the other hand, a direct application of assertion (a), we deduce that there exists an $\mathbb F$-predictable process $\varphi$ such that  $\theta=\varphi$ on $\Lbrack0,\tau\Lbrack$, and the $\mathbb G$-predictable in ${\cal X}^{\mathbb G}$ can be replaced with $\mathbb F$-predictable as well. Furthermore, for any $T\in (0,+\infty)$, any $c>0$ and any $\mathbb F$-predictable $H$ bounded by one, by putting $Q:=(G_T/E[G_T])\cdot P\sim P$ and $X^*_t:=\sup_{0\leq s\leq t}\vert X_s\vert$ for right continuous with left limits process $X$, we have 
\begin{eqnarray}\label{domination}
P\left( (H\theta I_{\{\vert\theta\vert\leq n}\is S^{\tau})_{T}^*>c\right)\geq Q\left((H\varphi I_{\{\vert\varphi\vert\leq n}\is S)_T^*>c\right)E[G_T].
\end{eqnarray}  
This allows us to conclude, due to \cite[Theorem 1.16, or Remark 2.2-(h)]{Stricker} again, that $\varphi\in L(S^T,\mathbb F)$, for any $T\in (0,+\infty)$. Thus, assertion (d) follows from combining this latter fact and \cite[Theorem 4]{StrickerSI}.\\
{\rm{(d)}} Let  $v$ be an $\mathbb F$-predictable process such that $v I_{\Lbrack0,\tau\Lbrack}\leq 0$ $P\otimes A$-a.e. This is equivalent to 
$$0=E[v^+\is A_{\tau\wedge T}]=E[v^+G_{-}\is A_{T}],$$
or equivalently $v^+=0$  $P\otimes A$-a.e.. This is obviously equivalent to  $v\leq 0$ $P\otimes A$-a.e., and assertion (d) is proved. This ends the proof of the lemma.\qed
\end{proof}
The following recalls $\mathbb G$-compensator of $\mathbb F$-optional process stopped at $\tau$. 
\begin{lemma}\label{lemmaV}
Let $V \in {\cal A}_{loc} ({\mathbb F})$, then we have $$(V^{\tau})^{p, {\mathbb G}}= I_{\Lbrack 0,\tau\Lbrack} G_-^{-1} \is ({\widetilde G} \is V)^{p, {\mathbb F}}.$$
 \end{lemma}
For the proof of this lemma and other related results, we refer to \cite{ACDJ1,ACDJ3}.
\section{Some useful martingale integrability properties}
The results of this section are new, very useful, and not technical at all.
\begin{lemma}\label{H0toH1martingales} Consider $K\in {\cal M}_{0,loc}(\mathbb H)$ with $1+\Delta K>0$, and let $H^{(0)}(K,P)$ be given by Definition \ref{Hellinger}. 
If $E[H^{(0)}_T(K,P)]<+\infty$, then $E[\sqrt{[K,K]_T}]<+\infty$ or equivalently $E[\displaystyle\sup_{0\leq t\leq T}\vert K_t\vert]<+\infty$.
\end{lemma}
\begin{proof}  Let $K\in {\cal M}_{0,loc}(\mathbb H)$ such that $1+\Delta K>0$ and $E[H^{(0)}_T(K,P)]<+\infty$. Then remark that, for any $\delta\in (0,1)$, we always have 
$$\Delta K-\ln(1+\Delta K)\geq {{\delta\vert \Delta K\vert }\over{\max\left(2(1-\delta),1+\delta^2\right)}}I_{\{\vert \Delta K\vert >\delta\}}+{{(\Delta K)^2}\over{1+\delta}}  I_{\{\vert \Delta K\vert \leq \delta\}}.$$
By combining this with (\ref{HellingerLog}), on the one hand, we deduce that
\begin{eqnarray*}
&&E\left[\langle K^c\rangle_T+\sum_{0<t\leq T}\vert \Delta K_t\vert  I_{\{\vert \Delta K_t\vert >\delta\}}+\sum_{0<t\leq T}(\Delta K_t)^2  I_{\{\vert \Delta K_t\vert \leq \delta\}}\right]\\
&&\leq C_{\delta} E\left[\langle K^c\rangle_T+\sum_{0<s\leq T}(\Delta K_s-\ln(1+\Delta K_s))\right] \leq 2C_{\delta}E[H^{(0)}_T(K, P)] < +\infty,\end{eqnarray*}
where $C_{\delta}:=\delta^{-1}+\max\left(\delta^{-1}-2,\delta\right)$.  On the other hand, it is clear that 
$$[K,K]^{1/2}_T\leq \sqrt{\langle K^c\rangle_T}+\sum_{0<t\leq T}\vert \Delta K_t\vert  I_{\{\vert \Delta K_t\vert >\delta\}}+\sqrt{\sum_{0<t\leq T}(\Delta K_t)^2  I_{\{\vert \Delta K_t\vert \leq \delta\}}}.$$
This ends the proof of the lemma.\qed
\end{proof}
\begin{proposition}\label{Hzero4Log(Z)}  Let $Z$ be a positive supermartingale such that $Z_0=1$. Then the following assertions hold.\\
{\rm{(a)}} There exist $K\in {\cal M}_{loc}(\mathbb H)$ and an nondecreasing and $\mathbb H$-predictable process $V$ such that $K_0=V_0=0$ , $\Delta K>-1$, and $Z={\cal E}(K)\exp(-V)$.\\
{\rm{(b)}} $-\ln(Z)$ is a uniformly integrable submartingale if and only if there exists a local martingale $N$ and a nondecreasing and predictable process $V$ such that $\Delta N>-1,\ Z={\cal E}(N)\exp(-V)$ and 
\begin{eqnarray}\label{Conditions}
E\left[V_T+H^{(0)}_T(N,P)\right]<+\infty.
\end{eqnarray}
{\rm{(c)}}  Suppose that there exist a finite sequence of positive supermartingale $(Z^{(i)})_{i=1,...,n}$ such that the product $Z:=\displaystyle\prod_{i=1}^n Z^{(i)}$ is a supermartingale. Then $-\ln(Z)$ is uniformly integrable submartingale if and only if all  $-\ln(  Z^{(i)})$, $i=1,...,n,$ are uniformly integrable submartingales.
\end{proposition}
\begin{proof} It is clear that assertion (a) is obvious. Thus, the rest of this proof will be given in two parts, where we prove assertions (b) and (c) respectively.\\
{\bf Part 1.} It is clear that there exist unique local martingale $N$ and a nondecreasing and predictable process $V$ such that  $N_0=V_0=0$, 
$$\Delta N>-1,\quad\quad Z={\cal E}(N)\exp(-V).$$
Thus, we derive 
\begin{eqnarray}\label{Ito}-\ln(Z)&&=-\ln({\cal E}(N))+V=-N+H^{(0)}(N,P)+ V,\end{eqnarray}
where both processes $V$ and $H^{(0)}(N,P)$ are nondecreasing. \\
Suppose that  $-\ln(Z)$ is a uniformly integrable submartingale, and let $(\tau_n)_n$ be a sequence of stopping times that increases to infinity and $N^{\tau_n}$ is a martingale. Then on the one hand, by stopping (\ref{Ito}) with $\tau_n$, and taking expectation afterwards  we get 
$$
E[-\ln(Z_{\tau_n\wedge T})]=E\left[V_{\tau_n\wedge T}+H^{(0)}_{\tau_n\wedge T}(N,P)\right].$$
On the other hand, since $\{ -\ln(Z_{\tau_n\wedge T}),\ n\geq 0\}$ is uniformly integrable and the RHS term of the above equality is increasing, by letting $n$ goes to infinity in this equality, (\ref{Conditions}) follows immediately. Now suppose that 
 (\ref{Conditions})  holds. As a consequence $E[H^{(0)}_{ T}(N,P) ]<+\infty$, and by combining this with Lemma \ref{H0toH1martingales} and (\ref{Ito}), we deduce that $-\ln(Z)$ is a uniformly integrable submartingale.\\
 {\bf Part 2.} Here we prove assertion (b). A direct application of assertion (a) to each $Z^{(i)}$ ($i=1,..., n$), we obtain the existence of $N^{(i)}\in {\cal M}_{loc}(\mathbb H)$ and nondecreasing and predictable $V^{(i)}$ such that
 $$\Delta N^{(i)}>-1,\quad\quad Z^{(i)}={\cal E}(N^{(i)})\exp(-V^{(i)}),\quad i=1,...,n.$$
 Furthermore, we derive 
 \begin{eqnarray*}-\ln(Z)&&=-\sum_{i=1}^n N^{(i)} +\sum_{i=1}^n H^{(0)}(N^{(i)},P)+ \sum_{i=1}^n V^{(i)}.\end{eqnarray*}
Hence, $-\ln(Z)$ is a uniformly  integrable submartingale if and only if 
\begin{eqnarray*}\label{Variable}
E\left[\sum_{i=1}^n H^{(0)}_T(N^{(i)},P)+ \sum_{i=1}^n V^{(i)}_T\right]<+\infty,\end{eqnarray*}
or equivalently $E\left[H^{(0)}_T(N^{(i)},P)+V^{(i)}_T\right]<+\infty$ for all $i=1,.., n$. Hence, thanks to assertion (b) ---applied to each $Z^{(i)}$ for $i=1,...,n$---, the proof of assertion (c) follows. This ends the proof of the proposition.\qed
\end{proof} 
\section{Martingales' parametrization via predictable characteristics}\label{sectionC}
Consider an arbitrary general model $(X, \mathbb H)$, and recall the corresponding notation given in the first paragraph of  Section \ref{section4}  up to (\ref{PCharac4X}).\\

For the  following, we refer to
\cite[Theorem 3.75]{J79}  and to \cite[Lemma 4.24]{JS03}.
\begin{theorem}\label{tmgviacharacteristics} Let $N\in {\cal M}_{0,loc}$. Then, there exist $\phi\in L^1_{loc}(X^c)$, $N'\in {\cal M}_{loc}$ with
$[N',X]=0$, $N'_0=0$ and functionals $f\in {\widetilde{{\cal P}}}$ and $g\in
{\widetilde{{\cal O}}}$ such that the following hold.\\
{\rm{(a)}} $\sqrt{(f-1)^2\star\mu}$ and $\Bigl(\sum (\widehat f- a)^2(1-a)^{-2} I_{\{a<1\}}I_{\{\Delta X=0\}}\Bigr)^{1/2}$ belong to ${\cal A}^+_{loc}$.\\
{\rm{(b)}}   $(g^2\star\mu)^{1/2}\in{\cal A}^+_{loc}$, $M^P_{\mu}(g\ |\ {\widetilde {{\cal P}}})=0,$ $P\otimes\mu$-a.e., $ \{a=1\}\subset\{\widehat f=1\}$, and
\begin{equation}
\label{Ndecomposition}
 N=\phi\cdot X^c+\left(f-1+{{\widehat f-a}\over{1-a}}I_{\{a<1\}}\right)\star(\mu-\nu)+g\star\mu+{N'}.
\end{equation} \end{theorem}
The quadruplet $ (\phi, f, g, N') $ is called throughout the paper by Jacod's components of $N$ (under $P$).
\section{A result on log-optimal portfolio: Choulli and Yansori (2020)}\label{DeepResultonDual}
 Herein, we consider the general setting and its notation, as in the first paragraph of Section \ref{section5}, where $(X,\mathbb H)$ is an arbitrary general model. 
\begin{theorem}\label{LemmaCrucial}
Let $X$ be an $\mathbb H$-semimartingale with predictable characteristics $\left(b,c,F, A\right)=\left(b^X,c^X,F^X, A^X\right)$, and ${\cal K}_{log}$ be the function given by (\ref{Kfunction}). 
Then the following assertions are equivalent.\\
{\rm{(a)}} The set ${\cal D}_{log}(X,\mathbb H)$, given by  (\ref{DeflatorsLOG}), is not empty (i.e. ${\cal D}_{log}(X,\mathbb H)\not=\emptyset$).\\
{\rm{(b)}} There exists an $\mathbb H$-predictable process $\widetilde\psi\in{\cal L}(X,\mathbb H)$ such that, for any $\varphi$ belonging to  ${\cal L}(X,\mathbb H)$,  the following hold 
\begin{eqnarray}
&&(\varphi-\widetilde\psi)^{tr}(b-c\widetilde\psi)+ \int \left( {{(\varphi-\widetilde\psi)^{tr}x}\over{1+{\widetilde\psi}^{tr}x}}-(\varphi-\widetilde\psi)^{tr}h(x)\right)F(dx)\leq 0, \label{C6forX}\\
&&E\left[{\widetilde V}^X_T+{1\over{2}}(\widetilde\psi^{tr}c\widetilde\psi\is A)_T+({\cal K}_{log}(\widetilde\psi^{tr}x)\star\nu)_T\right]<+\infty ,\label{Condi11}\\
&& {\widetilde V}^X:=\left[ \widetilde\psi^{tr}(b-c\widetilde\psi)+\int \left[{{\widetilde\psi^{tr}x }\over{1+\widetilde\psi^{tr}x}}-{\widetilde\psi}^{tr}h(x)\right] F(dx)\right]\is A\hskip 0.5cm\label{processV}
\end{eqnarray}
{\rm{(c)}} There exists a unique $\widetilde Z\in{\cal D}_{log}(X,\mathbb H)$ such that 
\begin{eqnarray}\label{dualSolution}
\inf_{Z\in{\cal D}(X,\mathbb H)}E[-\ln(Z_T)]=E[-\ln(\widetilde Z_T)].
\end{eqnarray}
{\rm{(d)}} There exists a unique $\widetilde\theta\in\Theta(X,\mathbb H)$ such that 
\begin{eqnarray}\label{PrimalSolution}
\sup_{\theta\in\Theta(X,\mathbb H)}E[\ln(1+(\theta\is X)_T)]=E[\ln(1+(\widetilde\theta\is X)_T)]<+\infty.
\end{eqnarray}
{\rm{(e)}} The num\'eraire portfolio exists , and its portfolio ``rate" $\widetilde\psi$ satisfies (\ref{Condi11}).\\
Furthermore, $\widetilde\theta (1+(\widetilde\theta\is X)_{-})^{-1}$ and $\widetilde\psi$ coincide $P\otimes A$-a.e.,  and 
\begin{eqnarray}
&& \widetilde\varphi\in L(X^c,\mathbb H)\cap {\cal L}(X,\mathbb H),\quad \sqrt{((1+\widetilde\varphi^{tr}x)^{-1}-1)^2\star\mu}\in{\cal A}^+_{loc}(\mathbb H),\label{integrabilities}\\
&&{1\over{\widetilde Z}}={\cal E}(\widetilde\psi\is X),\ \widetilde Z:={\cal E}(K^X){\cal E}(-V^X),\ K^X:=-\widetilde\psi\is X^c+{{-{\widetilde\Gamma}^X\widetilde\psi^{tr}x}\over{1+\widetilde\psi^{tr}x}}\star(\mu-\nu).\hskip 1cm\label{duality}\\
&&{\widetilde\Gamma}^X:=\left(1-a+\widehat{f^{(op)}}\right)^{-1},\quad f^{(op)}(t,x):=\left(1+{\widetilde\psi}^{tr}_tx\right)^{-1}.\label{Gammaf(op)}\end{eqnarray}
\end{theorem} 
\section{Proof of Lemma \ref{Lemma4Uniqueness}}\label{Section4Corollary5.3}
 Let $\theta_1$ and $\theta_2$ two elements of $ {\cal L}(S,\mathbb F)$ such that for any $\theta\in{\cal L}(S,\mathbb F)$, we have
\begin{eqnarray*}
&&(\theta-{\theta}_1)^{tr}(b-c{\theta}_1)+\int \left({{(\theta-{\theta}_1)^{tr}x}\over{1+{\theta}_1^{tr}x}}-(\theta-{\theta}_1)^{tr}h(x)\right)F(dx)\leq 0,\\
&&(\theta-{\theta}_2)^{tr}(b-c{\theta}_2)+\int \left({{(\theta-{\theta}_2)^{tr}x}\over{1+{\theta}_2^{tr}x}}-(\theta-{\theta}_2)^{tr}h(x)\right)F(dx)\leq 0.\end{eqnarray*}
By considering $\theta=\theta_2$ for the first inequality and $\theta=\theta_1$ for the second inequality and adding the resulting two inequalities  afterwards, we get 
\begin{eqnarray*}
(\theta_1-{\theta}_2)^{tr}c({\theta}_1-{\theta}_2)+\int\left({{1+{\theta}_2^{tr}x}\over{1+{\theta}_1^{tr}x}}+{{1+{\theta}_1^{tr}x}\over{1+{\theta}_2^{tr}x}}-2\right)F(dx)\leq 0.
\end{eqnarray*}
Then remark that, for any $x>0$, $x+x^{-1}-2$ is always nonnegative, and it is null if and only if $x=1$. Thus, by combining this fact with the above inequality we deduce that $c\theta_1=c\theta_2$ $P\otimes A$-a.e.. and $\theta_1^{tr}x=\theta_2^{tr}x$ $P\otimes A\otimes F$-a.e.. This ends the proof of the lemma. 
\begin{acknowledgements}This research is fully supported financially by the
Natural Sciences and Engineering Research Council of Canada, through Grant G121210818. \\ The authors would like to thank  Safa Alsheyab, Ferdoos Alharbi, Jun Deng, Youri Kabanov and Mich\`ele Vanmalele  for several comments, fruitful discussions on the topic, and/or for providing important and useful references.   
\end{acknowledgements}

\end{document}